\documentclass[letterpaper, 11pt]{article}

\usepackage[margin=1in]{geometry}
\usepackage{amsmath}
\usepackage{amssymb}
\usepackage{amsthm}
\usepackage{algorithm}
\usepackage{algorithmic}
\usepackage{multicol}
\usepackage{paralist}
\usepackage[usenames, dvipsnames]{color}
\usepackage[inline]{enumitem}
\usepackage{paralist}
\usepackage{bm}
\usepackage[normalem]{ulem}
\usepackage{pstricks,tikz}
\usepackage[colorlinks,urlcolor=blue,citecolor=blue,linkcolor=blue]{hyperref}

\usetikzlibrary{decorations.markings}
\usetikzlibrary{patterns}
\pgfdeclarelayer{background}
\pgfsetlayers{background,main}

\usepackage{graphicx,amsfonts,amsmath,amssymb,epsfig,color,multicol,pstricks,tikz, pgf}
\usetikzlibrary{decorations.markings}
\usetikzlibrary{patterns}
\usetikzlibrary{calc, intersections}
\usepackage{pgf,tikz}
\usetikzlibrary{calc}
\usepackage{tkz-euclide}
\usepackage{pgfplots}
\usetkzobj{all}

\pgfdeclarelayer{background}
\pgfsetlayers{background,main}

\tikzset{dotted pattern/.style args={#1}{
		postaction=decorate,
		decoration={
			markings,
			mark=between positions 0.25 and 0.75 step 0.25 with {
				\fill[radius=#1] (0,0) circle;
			}
		}
	},
	dotted pattern/.default={1pt},
}

		\tikzstyle{coiled}=[decoration={coil, aspect=0,,segment length=9pt}]
		
\tikzstyle{client}=[draw, circle, minimum size=1.5ex, inner sep=0,fill=black]
\tikzstyle{small}=[draw,circle,inner sep=0.05cm, fill=black]

\newcommand{\tikzAngleOfLine}{\tikz@AngleOfLine}
\def\tikz@AngleOfLine(#1)(#2)#3{%
	\pgfmathanglebetweenpoints{%
		\pgfpointanchor{#1}{center}}{%
		\pgfpointanchor{#2}{center}}
	\pgfmathsetmacro{#3}{\pgfmathresult}%
}

\newcommand{\drawBasic}[2]{
		\foreach \i in {1,2,...,3} {
			\node[client] (b1\i) at (-5.5,-2-\i/1.5) {} ;
			\node[client] (b2\i) at (-0.5,-2-\i/1.5) {} ;
		}
		\foreach \i in {1,2,...,3} {
			\node[client] (l\i) at (-4,\i-2) {} ;
			\node[client] (r\i) at (-2.5,\i-2) {} ;
		}

		\foreach \i in {1,2,3} {
			\node[client] (l'\i) at (-4.7,\i-1.5) {} ;
			\node[client] (r'\i) at (-1.8,\i-1.5) {} ;
		}

		\node[] (label) at (-4,2.3) {$A$} ;
		\node[] (label) at (-5,2.3) {$A'$} ;
		\node[] (label) at (-2,2.3) {$B$} ;
		\node[] (label) at (-1,2.3) {$B'$} ;

		\node[label={[xshift=0cm, yshift=-1cm]$S_1$}] (label) at (b13) {};
		\node[label={[xshift=0cm, yshift=-1cm]$S_2$}] (label) at (b23) {};

		\foreach \x in {l,r,b1} {
			\foreach \y in {b1,b2} {
				\foreach \i in {1,2,3} {
					\foreach \j in {1,2,3} {
						\ifthenelse {\equal{\x}{\y} }	{} {
							\draw[ultra thin, gray] (\x\i) -- (\y\j);								
						};
					}
				}
			}
		}
		\path[dotted pattern=0.5pt] (l1) -- (l2);
		\path[dotted pattern=0.5pt] (r1) -- (r2);
			
}

\newcommand{\drawBasicTriangles}[2]{
	\foreach \i in {1,2,3} {
		\node[client] (b1\i) at (-3.25+\i/4,-3.5) {} ;                        
	}
	\foreach \i in {1,2,3} {
		\node[client] (l\i) at (-3.5-\i/8,\i/2-3) {} ;
		\node[client] (r\i) at (-2+\i/8,\i/2-3) {} ;
	}
	
	\foreach \i in {1,2,3} {
		\node[client] (l'\i) at (-4+\i/8,\i/2) {} ;
		\node[client] (r'\i) at (-1.5-\i/8,\i/2) {} ;
	}
	
        \node[label={[xshift=-0.875cm, yshift=-0.5cm]$A$}] (label) at (l2) {};
        \node[label={[xshift=-0.875cm, yshift=0.5cm]$A'$}] (label) at (l'2) {};
        \node[label={[xshift=0.875cm, yshift=-0.5cm]$B$}] (label) at (r2) {};
        \node[label={[xshift=0.875cm, yshift=0.5cm]$B'$}] (label) at (r'2) {};
	
	\node[label={[xshift=0cm, yshift=-1cm]$S$}] (label) at (b12) {};

	\foreach \x in {l,r,b1} {
		\foreach \y in {b1} {
			\foreach \i in {1,2,3} {
				\foreach \j in {1,2,3} {
					\ifthenelse {\equal{\x}{\y} }	{} {
						\draw[ultra thin, gray] (\x\i) -- (\y\j);								
					};
				}
			}
		}
	}
	
}
\newcommand{\families}{
	\begin{center}
	\begin{tikzpicture}
	\begin{scope}[yscale=0.8,xscale=1.5, every node/.style={scale=0.8} ]
	
	\begin{scope}[shift={(0,0)}]
	\drawBasic{$H_1$}{$m=\Theta\paren{\binom{r-2}{2} M}$ edges \\ $C_r(G)= {\color{red} 0} \;/\;  {\color{blue} k \cdot M^{r/2}}$ $r$-cliques}
	\end{scope}
	
	\draw[blue, dashed, thick] (l2) -- (r1);
	\draw[blue, dashed, thick] (l'1) -- (r'2);
	\draw[red,thick, decorate, coiled] (l2) -- (l'1);
	\draw[red, thick, decorate, coiled] (r1) -- (r'2);
	
	\end{scope}

	\node[] (a) at (0,-4) {}; 
	\end{tikzpicture}
	\end{center}
}

\newcommand{\triangles}{
	\begin{center}
		\begin{tikzpicture}
		\begin{scope}[yscale=0.8,xscale=1.8, every node/.style={scale=0.8} ]
		
		\begin{scope}[shift={(0,0)}]
		\drawBasicTriangles{$H_1$}{$m=\Theta\paren{M}$ edges \\ $C_3(G)= {\color{red} 0} \;/\;  {\color{blue} k \cdot M^{3/2}}$}
		\end{scope}
		
		\draw[blue, dashed, thick] (l2) -- (r1);
		\draw[blue, dashed, thick] (l'1) -- (r'2);
		\draw[red,thick, decorate, coiled] (l2) -- (l'1);
		\draw[red, thick, decorate, coiled] (r1) -- (r'2);
		
		\end{scope}
		
		\node[] (a) at (0,-4) {}; 
		\end{tikzpicture}
	\end{center}
}

\newcommand{\lowerBoundItemOne}{
	\coordinate[] (A) at (0,0);
	\coordinate[] (B) at (0,4);
	\coordinate[] (X) at (4,6);
	\coordinate[] (Y) at (4,-2);
	
	\begin{scope}
		\path[clip] (A) -- (X) -- (Y);
		\fill[red, opacity=0.2, draw=black] (A) circle (5mm);
		\node at ($(A)+(10:11.5mm)$) {$M_s^{1/s}$};
		
	\end{scope}
	
	\begin{scope}
		\path[clip] (B) -- (X) -- (Y);
		\fill[red, opacity=0.5, draw=black] (B) circle (5mm);
		\node at ($(B)+(-1:11.5mm)$) {$M_s^{1/s}$};
		
	\end{scope}
	
	\tikzstyle{cycle}=[draw,circle,inner sep=0.05cm, fill=black]
	
	\foreach \i in {0,2,3,4} { 
		\node[cycle] at (0,\i) {};	
	}
	\path[dotted pattern=1.0pt] (0,0) -- (0,2);
	\path[dotted pattern=1.0pt] (3,-1) -- (3,4);
	\path[dotted pattern=1.0pt] (6,-1) -- (6,4);

	\foreach \i in {-2,-1,4,5,6} { 
		\node[cycle] at (3,\i) {};				
		\node[cycle] at (6,\i) {};					
		\foreach \j in {0,3,4} { 	
			\ifthenelse {\i = \j}	{ } {
				\draw[ultra thin] (0,\j) -- (3,\i);
				\draw[ultra thin, dashed] (0,\j) -- (6,\i);	
			};
		} 
	}
	
	\node[label={$S$}] (S) at (0,-3.5) {};
	\node[label={$V \setminus S$}] (S) at (4.5,-3.5) {};
	\draw [thick, decoration={ brace, raise=0.5cm,amplitude=10pt}, decorate] (0,0) -- (0,4) node [midway, left=10pt+10mm] {$c$} ;
	\draw [thick, decoration={ brace, raise=0.5cm,amplitude=10pt}, decorate] (6,6) -- (6,-2) node [midway, left=-10pt-20mm] {$n-c$} ;

	\foreach \i in {-2,-1,4,5,6} { 
		\draw (3,\i) coordinate (Z) -- ++(18.5:1cm) coordinate (W);
		\draw (3,\i) coordinate (Z) -- ++(-18.5:1cm) coordinate (R);
		\begin{scope}
			\path[clip] (Z) -- (W) -- (R);
			\fill[green, opacity=0.5, draw=black] (Z) circle (5mm);
			\node at ($(Z)+(0:7mm)$) {$d$};		
		\end{scope}
		
		\draw (6,\i) coordinate (Z) -- ++(18.5:-1cm) coordinate (W);
		\draw (6,\i) coordinate (Z) -- ++(-18.5:-1cm) coordinate (R);
		\begin{scope}
			\path[clip] (Z) -- (W) -- (R);
			\fill[green, opacity=0.5, draw=black] (Z) circle (5mm);
			\node at ($(Z)+(0:-7mm)$) {$d$};		
		\end{scope}
		
	}
}

\newcommand{\lowerBoundItemTwo}{
	\coordinate[] (A) at (0,0);
	\coordinate[] (B) at (0,4);
	\coordinate[] (X) at (4,6);
	\coordinate[] (Y) at (4,-2);
	
	\begin{scope}
		\path[clip] (A) -- (X) -- (Y);
		\fill[red, opacity=0.2, draw=black] (A) circle (5mm);
		\node at ($(A)+(10:11.5mm)$) {$n-b$};
		
	\end{scope}
	
	\begin{scope}
		\path[clip] (B) -- (X) -- (Y);
		\fill[red, opacity=0.5, draw=black] (B) circle (5mm);
		\node at ($(B)+(-1:11.5mm)$) {$n-b$};
		
	\end{scope}
	
	\tikzstyle{cycle}=[draw,circle,inner sep=0.05cm, fill=black]
	
	\foreach \i in {0,2,3} { 
		\node[cycle] at (0,\i) {};	
	}
	\path[dotted pattern=1.0pt] (0,0) -- (0,2);
	\path[dotted pattern=1.0pt] (3,-1) -- (3,4);
	\path[dotted pattern=1.0pt] (6,-1) -- (6,4);

	\foreach \i in {-2,-1,4,5,6} { 
		\node[cycle] at (3,\i) {};				
		\node[cycle] at (6,\i) {};					
		\foreach \j in {0,3,4} { 	
			\ifthenelse {\i = \j}	{ } {
				\draw[ultra thin] (0,\j) -- (3,\i);
				\draw[ultra thin, dashed] (0,\j) -- (6,\i);	
			};
		} 
	}
	
	\node[label={$S$}] (S) at (0,-3.5) {};
	\node[label={$S \setminus V$}] (S) at (4.5,-3.5) {};
	\draw [thick, decoration={ brace, raise=0.5cm,amplitude=10pt}, decorate] (0,0) -- (0,4) node [midway, left=10pt+10mm] {$b$} ;
	\draw [thick, decoration={ brace, raise=0.5cm,amplitude=10pt}, decorate] (6,6) -- (6,-2) node [midway, left=-10pt-20mm] {$n-b$} ;

	\foreach \i in {-2,-1,4,5,6} { 
		\draw (3,\i) coordinate (Z) -- ++(18.5:1cm) coordinate (W);
		\draw (3,\i) coordinate (Z) -- ++(-18.5:1cm) coordinate (R);
		\begin{scope}
			\path[clip] (Z) -- (W) -- (R);
			\fill[green, opacity=0.5, draw=black] (Z) circle (5mm);
			\node at ($(Z)+(0:7mm)$) {$d$};		
		\end{scope}
		
		\draw (6,\i) coordinate (Z) -- ++(18.5:-1cm) coordinate (W);
		\draw (6,\i) coordinate (Z) -- ++(-18.5:-1cm) coordinate (R);
		\begin{scope}
			\path[clip] (Z) -- (W) -- (R);
			\fill[green, opacity=0.5, draw=black] (Z) circle (5mm);
			\node at ($(Z)+(0:-7mm)$) {$d$};		
		\end{scope}
		
	}
}

\newcommand{\MomentsLB}{
	\begin{center}
		\begin{tikzpicture}
		\begin{scope}[scale=1.5]
		\begin{scope}[scale=0.7]
		
		\draw  (6.5,-1) circle [x radius=1cm, y radius=1.2cm ] node[xshift=-1.4cm] (R) {$R$};
		
		\draw [thick, decoration={ brace, raise=0.5cm,amplitude=10pt}, decorate] (7,.1) -- (7.2,-2) node [midway, right, yshift=0cm, xshift=1cm] {$\alpha$} ;

		\begin{scope}[shift={(6.5,-1)},scale=0.15] 
		\foreach \x in {1,...,6}{
			\pgfmathparse{(\x-1)*60+floor(\x/9)*22.5}
			\node[draw,circle,inner sep=0.05cm, fill=black] (N-\x) at (\pgfmathresult:5.4cm) {};
		}
		\foreach \x [count=\xi from 1] in {1,...,6}{
			\foreach \y in {\x,...,6}{
				\path[dashed, thin] (N-\xi) edge[-] (N-\y);
			}
		}
		\end{scope}
		\end{scope}	
		
		\foreach \i\y in {-2/1,0.5/2,3/3} {
			\draw  (\i,-3) circle [x radius=.7cm, y radius=.6cm ] coordinate (w\i);
			\begin{scope}[shift={(\i,-3)}, xscale=.15] 
			\foreach \x\z in {1/1,2.5/2,4/3,5.5/4,7/5}{
				\node[draw,circle,inner sep=0.05cm, fill=black] (w\y\z) at (\x-3.7,0) {};
			}
			\end{scope}
		}
		\node[] (a) at (0.5,0)  {};
		\foreach \i\j in {-1.5/1,-1/2,-.5/3,0/4,.5/5,1.5/6,2/7,2.5/8} {
			\node[draw,circle,inner sep=0.05cm, fill=black]  (a\j) at (\i,0) {}; 
			\node[draw,circle,inner sep=0.05cm, fill=black] (b\j) at (\i,-1) 
			{};			
		}
		\path[dotted pattern] (.7,0) -- (1.3,0);
		\path[dotted pattern] (.7,-1) -- (1.3,-1);

		\foreach \i in {-1,-.5,0,0.5,1.5,2} {
			\draw[dashed, red, thin] (\i,0) -- (w23); 
			\draw[dashed, red, thin] (\i,-1) -- (w23); 
			
		}
	
		\coordinate (A) at (w23);
		\coordinate (B) at ($(A)+(57:1)$);
		\coordinate (C) at ($(A)+(125:1)$);
		\tkzMarkAngle[fill=red,%
		size=0.55, opacity=0.2, draw=black](B,A,C);
		 \coordinate (label) at ($(w23)+(85:.3)$);
		 \draw (label) -- ($(label)+(.8,0.4)$) node[right] {$M_s^{1/s}$} ;
		
		\coordinate (mid) at (2,-1);
		\tikzAngleOfLine(w21)(mid){\AngleStart};
		\tikzAngleOfLine(w25)(mid){\AngleEnd};
		\coordinate (left) at ($(mid)+(\AngleStart:-.5cm)$);
		\coordinate (right) at ($(mid)+(\AngleEnd:-.5cm)$);
		\coordinate (label2) at ($(mid)+(-110:.4cm)$);
		\draw[] (label2) -- ( $(label2)+(.5,-.3)$) node[right] {$\ell$};
		\tkzMarkAngle[fill=red,	size=0.6, opacity=0.2, draw=black](left,mid,right);
		
		\draw[blue, thin, dashed] (mid) -- ($(mid)+(101:1cm)$);				
		\draw[blue, thin, dashed] (mid) -- ($(mid)+(79:1cm)$);				
		
		\coordinate (l1) at ($(mid)+(101:.45cm)$);
		\coordinate (r1) at ($(mid)+(79:.45cm)$);				
		\tkzMarkAngle[fill=blue, size=0.55, opacity=0.2, draw=black](r1,mid,l1);
		\draw[]  ($(mid) + (0,0.5)$) -- ( $(label2)+(.6,-.15)$) ;
		
		\draw[black, thin, dashed] (mid) -- ($(mid)+(120:1cm)$);
		\draw[black, thin, dashed] (mid) -- ($(mid)+(60:1cm)$);				
		\coordinate (l2) at ($(mid)+(120:.5cm)$);
		\coordinate (r2) at ($(mid)+(60:.5cm)$);				
		\tkzMarkAngle[fill=black, size=0.4, opacity=0.2, draw=black](r2,mid,l2);
		\draw[]  ($(mid) + (0.1,0.2)$) -- ($(mid)+(0.9,-0.5)$)  node[below]{$d$};

		\path[dotted pattern=1pt] (-1.25,-3) -- (-.25,-3);
		\path[dotted pattern=1pt] (1.25,-3) -- (2.25,-3);
		\node[] (w1_label) at (-2,-4) {$W_1$};
		\node[] (w1_label) at (0.5,-4) {$W_j$};
		\node[] (w1_label) at (4,-4) {$W_{d/\ell}$};
		\draw [thick, decoration={ brace, raise=0.5cm,amplitude=10pt}, decorate] (-2.5,-3.6) -- (-2.5,-2.3) node [midway, right, yshift=0cm, xshift=-1.8cm] {$c$} ;
		
		\draw [thick, decoration={ brace, raise=0.5cm,amplitude=10pt}, decorate] (-2.5,-1.2) -- (-2.5,0.2) node [midway, right, yshift=0cm, xshift=-1.8cm] {$2n$} ;
		
		\node[] (labelAUB) at (-2.2,-.5) {$A\cup B$};
		\end{scope}
		\end{tikzpicture}
	\end{center}
	
}

\newcommand{\drawNbrs}{
\begin{center}
\begin{tikzpicture}

	\node[draw,circle,inner sep=0.05cm, fill=black, label={above:$a_i$}] (a) at (0,0) {};

	\foreach \y\i in {0/1,-1/2,-2/3,-4/4,} {
		\node[draw,circle,inner sep=0.05cm, fill=black, 
		] (b\i) at (3,\y) {};
	}
	\node[circle,inner sep=.2cm, fill=white, 
	] (b0) at (3,1) {};
	\path[dotted pattern] (3,0.2) -- (3,0.8);
	\path[dotted pattern] (3,-.2) -- (3,-.8);
	\path[dotted pattern] (3,-1.2) -- (3,-1.8);
	\path[dotted pattern] (3,-2.4) -- (3,-3.6);
	\path[dotted pattern] (3,-4.2) -- (3,-4.8);
	
	\foreach \y in {0,-0.2, ...,-.8} {
		\draw[ultra thin, black] (a) -- (3,\y);
		\draw[thin, blue, dashed] (a) -- (3,\y-1);
		\draw[ultra thin, black] (a) -- (3,\y-2);
		\draw[ultra thin, black] (a) -- (3,\y-3);
	}
	\draw[ultra thin, black] (a) -- (3,-4);
	\node[right,xshift=1mm] at (b1) {$b_i$};
	\node[right,xshift=1mm] at (b2) {$b_{i+(j-1)\cdot \ell}$};
	\node[right,xshift=1mm] at (b3) {$b_{i+j \cdot \ell-1}$};
	\node[right,xshift=1mm] at (b4) {$b_{i+d}$};
	\draw [thick, decoration={ brace, raise=0.5cm,amplitude=10pt}, decorate] (4.4,-.8) -- (4.4,-2.2) node [text width=3cm,align=center,midway, right, yshift=0cm, xshift=.8cm] {The $j\th$ block \\ of potential neighbors} ;

	\foreach \y\W in {2/1,-.5/2,-3.5/3} {
		\coordinate (W\W) at (-3,\y-1);
		\draw  (W\W) circle  [x radius=.7cm, y radius=.9cm] node{}; 
		\foreach \i\j in {.3/1,0/2,-.3/3,-.6/4} {
			\node[draw,circle,inner sep=0.05cm, fill=black] (w_\W_\j) at (-3,\y+\i-.9) {};				
		}
	}
	\node[xshift=-1.3cm] at (W1) {$W_1$};
	\node[xshift=-1.3cm] at (W2) {$W_j$};
	\node[xshift=-1.3cm] at (W3) {$W_{d/\ell}$};
	\path[dotted pattern={.5pt}] (-3,.8) -- (-3,.5);
	\path[dotted pattern] (-3,-2.5) -- (-3,-3.5);

	\foreach \i in {1,2,3} {
		\draw[thin, dashed, red] (a) -- (w_2_\i);
	}

\coordinate (w_2_0) at ($(w_2_1)+(0,.2)$);
	\tkzMarkAngle[fill=black, size=1, opacity=0.2, draw=black](b4,a,b1);
	\tkzLabelAngle[pos = 0.6](b3,a,b1){}
	\draw[black, thick] ($(a)+(.6,-.1)$) -- ($(a)+(.6,.1)$) node[above] {$d$} ;
	
	\tkzMarkAngle[fill=blue, size=1.3, opacity=0.2, draw=black](b3,a,b2);
	\tkzLabelAngle[pos = 1.2](b3,a,b2){\bm{$\ell$}}
	\draw[black, thick] ($(a)+(.6,-.1)-(1.2,0.2)$) -- ($(a)+(.6,.1)-(1.2,0)$) node[above] {$\ell$} ;
	
	\tkzMarkAngle[fill=red, size=1, opacity=0.2, draw=black](w_2_0,a,w_2_3);
	\tkzLabelAngle[pos = .7](w_2_1,a,w_2_3){}

\end{tikzpicture}
\end{center}
}

\newcommand{\kinter}{{$k$-intersection}}

\theoremstyle{plain}
\newtheorem{thm}{Theorem}[section]

\newtheorem{cor}[thm]{Corollary}

\newtheorem{claim}[thm]{Claim}

\theoremstyle{definition}
\newtheorem{dfn}[thm]{Definition}

\theoremstyle{remark}
\newtheorem{rem}[thm]{Remark}

\newtheorem{que}[thm]{Question}

\newcommand{\dft}[1]{\textbf{\textit{#1}}}

\newcommand{\calE}{\mathcal{E}}
\newcommand{\calG}{\mathcal{G}}

\newcommand{\wtM}{\widetilde{M}}
\newcommand{\mG}{\mathcal{G}}

\DeclareMathOperator*{\E}{\mathbf{E}}
\newcommand{\R}{\mathbf{R}}
\newcommand{\N}{\mathbf{N}}

\DeclareMathOperator{\cost}{cost}
\DeclareMathOperator{\disj}{disj}
\DeclareMathOperator{\inter}{int}

\newcommand{\gmom}{g_M}

\newcommand{\e}{\varepsilon}

\newcommand{\abs}[1]{\left|#1\right|}

\newcommand{\paren}[1]{\left(#1\right)}
\newcommand{\set}[1]{\left\{#1\right\}}
\newcommand{\sucht}{\,\middle|\,}
\renewcommand{\th}{^{\textrm{th}}}

\DeclareMathOperator{\tvd}{dist_{TV}}
\DeclareMathOperator{\neighbor}{nbr}
\DeclareMathOperator{\nbr}{nbr}
\DeclareMathOperator{\degree}{d}
\DeclareMathOperator{\pair}{pair}

\newcommand{\mA}{\mathcal{A}}

\newcommand{\mhat}{\widehat{m}}

\def\withcolors{1}
\ifnum\withcolors=1
\newcommand{\tnote}[1]{{\color{red}{#1}}}

\newcommand{\strike}[1]{{\sout{#1}}}

\else
\newcommand{\tnote}[1]{{{#1}}}

\newcommand{\strike}[1]{{{}}}

\fi

\title{Lower Bounds for Approximating Graph Parameters via Communication Complexity}
\date{}
\author{%
  \begin{tabular}{c@{\extracolsep{16pt}}c}
    Talya Eden\thanks{\texttt{talyaa01@gmail.com}} & Will Rosenbaum\thanks{
      \texttt{will.rosenbaum@gmail.com}
    }
    \\
    \multicolumn{2}{c}{School of Electrical Engineering}\\
    \multicolumn{2}{c}{Tel Aviv University}\\ 
    \multicolumn{2}{c}{Tel Aviv 6997801}\\ 
    \multicolumn{2}{c}{Israel}
  \end{tabular}
}

\begin{document}
\maketitle

\begin{abstract}
  In a celebrated work, Blais, Brody, and Matulef~\cite{Blais2012} developed a technique for proving property testing lower bounds via reductions from communication complexity. Their work focused on testing properties of functions, and yielded new lower bounds as well as simplified analyses of known lower bounds. Here, we take a further step in generalizing the methodology of~\cite{Blais2012} to analyze the query complexity of graph parameter estimation problems. In particular, our technique decouples the lower bound arguments from the representation of the graph, allowing it to work with any query type.

  We illustrate our technique by providing new simpler proofs of previously known tight lower bounds for the query complexity of several graph problems: estimating the number of edges in a graph, sampling edges from an almost-uniform distribution, estimating the number of triangles (and more generally, $r$-cliques) in a graph, and estimating the moments of the degree distribution of a graph. We also prove new lower bounds for estimating the edge connectivity of a graph and estimating the number of instances of any fixed subgraph in a graph. We show that the lower bounds for estimating the number of triangles and edge connectivity also hold in a strictly stronger computational model that allows access to uniformly random edge samples.
\end{abstract}

\section{Introduction}

Since the seminal work of Yao~\cite{Yao1979}, (two party) communication complexity has become a central topic in computational complexity. While communication complexity is well-studied in its own right, its importance in complexity theory is highlighted by its numerous applications, for example, in proving lower bounds for Turing machines~\cite{Paturi1986, Kalyanasundaram1992-1}, streaming algorithms~\cite{Alon1996}, circuit complexity~\cite{Karchmer1990}, distributed algorithms~\cite{Hromkovic2013, DasSarma2012}, and algorithmic game theory~\cite{Gonczarowski2015}. Lower bounds via reductions from communication complexity also tend to be conceptually simple: delicate analysis is delegated to a small handful of fundamental results in communication complexity.

In typical applications of communication complexity, instances of the problem considered can be readily partitioned into two or more pieces. This not the case in the context of property testing~\cite{Goldreich1998}, where the goal is to distinguish instances of a problem that satisfy some property from those that are ``far'' from satisfying the property. Thus, the work of Blais, Brody, and Matulef~\cite{Blais2012} was surprising, as it drew a close connection between communication complexity and property testing. Specifically, in~\cite{Blais2012}, the authors develop a framework for applying communication complexity to obtain lower bounds in the property testing model. Their methodology yielded new results in property testing, as well as simpler proofs of known results.

The framework of~\cite{Blais2012} (which was subsequently generalized by Goldreich in~\cite{Gol13_CC}) applied to the property testing of functions. In this context, an algorithm is given query access to a function, where each query consists of evaluating the function a single input specified by the algorithm. The goal is to design algorithms that distinguish functions having some property---for example being monotonic, or $k$-linear---from those that are far from having the property\tnote{,} in the sense that a constant fraction of the function's outputs must be changed in order for it to satisfy the property. The basic methodology of~\cite{Blais2012} and~\cite{Gol13_CC} is as follows. Let $P$ be a property, $\Pi$ a two-party communication problem, and $x$ and $y$ the private inputs of the two parties. The idea is then to construct a function $f_{x, y}$ with the following properties:
\begin{enumerate}
\item if $\Pi(x, y) = 1$, then $f_{x, y}$ satisfies $P$;
\item if $\Pi(x, y) = 0$, then $f_{x, y}$ is far from satisfying $P$;
\item for each $z$ in the domain of $f_{x, y}$, $f_{x, y}(z)$ can be computed from $x$ and $y$ using at most $B$ bits of communication.
\end{enumerate}
The main result of~\cite{Blais2012} (cf.~\cite{Gol13_CC}) states that under the three conditions above, any algorithm that tests $P$ requires at least $R(\Pi) / B$ queries, where $R(\Pi)$ is the randomized communication complexity of $\Pi$.

In this work, our goal is to prove lower bounds on the number of queries necessary to (approximately) solve various graph problems. In particular, we adapt the framework described above to the context of graph parameter estimation problems. Any graph $G = (V, E)$ on $n$ nodes can be viewed as a Boolean function whose values are are the entries of the adjacency matrix of $G$. Thus we can directly apply the methodology of~\cite{Blais2012}. This view of graph property testing corresponds to the ``dense graph'' model introduced in~\cite{Goldreich1998}: the graph is accessed only through ``pair queries'' (i.e., asking if two nodes share an edge), and two graphs are far apart only if they differ on $\e n^2$ edges. However, this model is not suitable for analyzing graphs with $o(n^2)$ edges. In order to deal with property testing and parameter estimation in non-dense families of graphs, more refined graph models were introduced. These models allow additional types of queries---degree queries and neighbor queries---that cannot be handled in any obvious way using the models of~\cite{Blais2012} and~\cite{Gol13_CC}.

\subsection{Our Results}

Our main result (Theorem~\ref{thm:query-lb}) gives a general reduction from communication complexity problems to graph problems. The theorem is closely related to Theorem~3.1 in~\cite{Gol13_CC}, but it makes a further step at generalizing the results of~\cite{Blais2012}. In particular, our result makes no assumptions about the representation of the objects in question or the types of queries allowed (although nontrivial lower bounds are only obtained when the allowable queries can be efficiently simulated with a 2-party communication protocol). Thus, we believe our main result offers several advantages. Since our approach decouples the graph queries from the representation of the graph (e.g., by its adjacency lists or adjacency matrix), it can handle many different query models (possibly simultaneously). Further, our framework may be useful in distinguishing the relative power of different query access models. Finally, we believe our results unify and simplify previous lower bound arguments for graph parameter estimation (which typically relied upon careful analysis of statistical distances between families of graphs).

We apply our lower bound framework to the following graph problems that have been addressed previously:
\begin{enumerate}
\item estimating the number of edges~\cite{feige06, Goldreich2008} (Section~\ref{sec:subgraph-counting}, Appendix~\ref{sec:edge-counting}),
\item sampling edges from an almost-uniform distribution~\cite{Eden2017} (Section~\ref{sec:edge-sampling}),
\item estimating the number of $r$-cliques~\cite{CLIQUES} (in particular, triangles~\cite{Eden2015}; Appendix~\ref{sec:cliques} and Section~\ref{sec:triangles}, respectively)
\item estimating the moments of the degree distribution~\cite{DegMoments} (Appendix~\ref{sec:moments}).
\end{enumerate}
Our technique yields tight lower bounds for each of the problem listed above. Further, we prove new lower bounds for:
\begin{enumerate}
\item[5.] estimating the number of instances of any fixed subgraph $H$ (Section~\ref{sec:subgraph-counting}),
\item[6.] estimating the edge-connectivity of a graph (Section~\ref{sec:connectivity}).
\end{enumerate}
In results (3) and (6), our lower bounds hold in a strictly stronger graph access model that additionally allows uniformly random edge samples. Interestingly, all of the lower bounds we prove are polynomial in the size of the instance, whereas the lower bounds presented in~\cite{Blais2012} and~\cite{Gol13_CC} are typically logarithmic in the instance size.



\subsection{Related Work}

A model for query-based sublinear graph algorithms was first presented in the seminal work of Goldreich, Goldwasser, and Ron~\cite{Goldreich1998} in the context of property testing. Their model is appropriate for dense graphs, as only ``pair queries'' (i.e., queries of the form \emph{``Do vertices $u$ and $v$ share an edge?''}) are allowed. An analogous model for bounded degree graphs was introduced by Goldreich and Ron in~\cite{GR02}. In this model, it is assumed that all vertices have degree at most $\Delta$, and the graph is accessed via neighbor queries (\emph{``Who is $v$'s $i\th$ neighbor?''} for $i \leq \Delta$). A similar model for sparse graphs was introduced by Parnas and Ron in~\cite{Parnas2002}, which does not assume that the maximum degree in the graph is bounded, and additionally allows degree queries (\emph{``What is $v$'s degree?''}). Kaufman, Krivelevich and Ron~\cite{KKR04} introduced the general graph model which allows all of the above queries---pair, degree and neighbor queries. The lower bounds we prove all apply to the general graph model.


The problem of estimating the average degree of a graph (or equivalently, the number of edges in a graph) was first studied by Feige~\cite{feige06}. In~\cite{feige06}, Feige proves tight bounds on the number of degree queries necessary to estimate the average degree. In~\cite{Goldreich2008}, Goldreich and Ron study the same problem, but in a model that additionally allows neighbor queries. In this model, they prove matching upper and lower bounds for the number of queries needed to estimate the average degree. In Corollary~\ref{cor:edge-counting} we achieve the same lower bound as~\cite{Goldreich2008} for estimating the number of edges in a graph (their lower bound as well as ours also holds when allowing for pair queries). The related problem of sampling edges from an almost-uniform distribution was recently studied by Eden and Rosenbaum in~\cite{Eden2017}. They prove tight bounds on the number of queries necessary to sample an edge in a graph from an almost-uniform distribution. In Theorem~\ref{thm:edge-sampling} we present a new derivation of the lower bound presented in~\cite{Eden2017}.

Eden et al.\ \cite{Eden2015} prove tight bounds on the number of queries needed to estimate the number of triangles. Their results were generalized by Eden et al.\ in~\cite{CLIQUES} to approximating the number of $r$-cliques for any $r \geq 3$. In Corollary~\ref{cor:cliques} and Theorem~\ref{thm:triangles}, we present a new derivation of the lower bound of~\cite{Eden2015} for estimating the number of triangles. In Appendix~\ref{sec:cliques} (Theorem~\ref{thm:cliques_second}), we generalize the triangle lower bound construction to obtain a lower bound for $r$-cliques matching that of~\cite{CLIQUES}.

In~\cite{GRS11}, Gonen et al.\ study the problem of approximating the number of $s$-star subgraphs, and give tight bounds on the number of (degree, neighbor, pair) queries needed to solve this problem. As noted by Eden et al.~\cite{DegMoments}, counting $s$-stars is closely related to computing the $s\th$ moment of the degree sequence. In~\cite{DegMoments}, the authors provide a simpler optimal algorithm for computing the $s\th$ moment of the degree sequence that has better or matching query complexity when the algorithm is also given an upper bound on the arboricity of the graph. In Theorems~\ref{thm:moments_1term} and~\ref{thm:moments_2term} we prove the lower bounds of~\cite{DegMoments}, which build upon and generalize the lower bounds of~\cite{GRS11}.

The recent paper of Aliakbarpour et al.~\cite{Aliak} proposes an algorithm for estimating the number of $s$-star subgraphs that is allowed uniformly random edge samples as a basic query as well as degree queries. Interestingly, the additional computational power afforded by random edge queries allows their algorithm to break the lower bound of~\cite{GRS11}. We remark that our lower bounds for estimating the number of triangles (Theorem~\ref{thm:triangles}) and edge connectivity (Theorem~\ref{thm:connectivity}) still hold in this stronger query access model.

\section{Preliminaries}
\label{sec:preliminaries}

\subsection{Graph Query Models}
\label{sec:query-model}
Let $G = (V, E)$ be a graph where $n = \abs{V}$ is the number of vertices and $m = \abs{E}$ is the number of edges. We assume that the vertices $V$ are given distinct labels, say, from $[n] = \set{1, 2, \ldots, n}$. For $v \in V$, let $\Gamma(v)$ denote the set of neighbors of $v$, and $\deg(v) = \abs{\Gamma(v)}$ is $v$'s degree. For each $v \in V$, we assume that $\Gamma(v)$ is ordered by specifying some arbitrary bijection $\Gamma(v) \to [\deg(v)]$ so that we may refer unambiguously to $v$'s $i\th$ neighbor. We let $\calG_n$ denote the set of all graphs on $n$ vertices, together will all possible labelings of the vertices (from $[n]$) and all orderings of the neighbors of each vertex, and we define $\calG = \bigcup_{n \in \N} \calG_n$.

We consider algorithms that access $G$ via queries. In general, a \dft{query} is an arbitrary function $q : \calG \to \set{0, 1}^*$. We are interested in the following question: \emph{``Given a set $Q$ of allowable queries and a graph problem $g$ (e.g. a computing function, estimating a graph parameter, etc.), how many queries $q \in Q$ are necessary to compute $g$?''}

Since we associate the vertex set $V$ with the set $[n]$, we allow algorithms to have free access to the vertex set of the graph\footnote{In the sparse and general graph models of property testing, it is often assumed that the identities of the vertices are not known in advance. Rather, algorithms may sample vertices from a uniform distribution or discover new vertices that are neighbors of known vertices. Since the current paper aims to prove lower bounds, the assumption that the identities of vertices are known in advance is harmless.}---algorithms are only charged for obtaining information about the edges of a graph. We focus on models that allow the following types of queries:
\begin{enumerate}
\item \dft{degree query} $\degree : V \to [n - 1]$, where $\degree(v)$ returns $v$'s degree.
\item \dft{neighbor query} $\neighbor_i : V \to V \cup \set{\varnothing}$ for $i\in[n-1]$, where $\neighbor_i(v)$ returns $v$'s $i\th$ neighbor if $i \leq \deg(v)$ and $\varnothing$ otherwise.
\item \dft{pair query} $\pair : V \times V \to \set{0, 1}$, where $\pair(u, v)$ returns $1$ if $(u, v) \in E$ and $0$ otherwise $(u, v) \notin E$. 
\end{enumerate}
Taking $Q$ to be the set of all neighbor, degree, and pair queries, we have $\abs{Q} = O(n^2)$. This query model is known as the \dft{general graph model} introduced in~\cite{KKR04}.\footnote{Allowing only neighbor queries while assuming the graphs has maximal degree $d$, and considering the distance with respect to $n\cdot d$ is known as the ``bounded degree graph'' model, while only allowing pair queries and considering the distance with respect to $n^2$ is the ``dense graph'' model.} In Theorems~\ref{thm:triangles} and~\ref{thm:connectivity}, we also consider an expanded model that allows random edges to be sampled from a uniform distribution (cf.~\cite{Aliak}).

We wish to characterize the \dft{query complexity} of graph problems, that is, the minimum number of queries necessary to solve the problem. We consider randomized algorithms, and we assume the randomness is provided via a random string $\rho \in \set{0, 1}^\N$. Since the query complexity of the various estimation problems we consider depends on the measure being estimated, we use the expected query complexity, rather than the worst case query complexity.\footnote{In communication complexity, it is customary to use worst-case complexity when analyzing randomized protocols. This is done without loss of generality, as protocol with a given expected communication cost can be converted to a protocol with asymptotically equal worst-case communication cost and slightly higher error probability. Such a transformation from expected to worst-case lower bounds is not generally possible for query complexity when the expected cost of a protocol depends on the parameter being estimated.}

Most of our results are lower bounds on the number of expected queries necessary to estimate graph parameters.

\begin{dfn}
  \label{dfn:graph-param}
  A \dft{graph parameter} is a function $g : \calG \to \R$ that is invariant under any permutation of the vertices of each $G \in \calG$. Formally, $g$ is a graph parameter if for every $n \in \N$, $G = (V, E) \in \calG_n$ and every permutation  $\pi : [n] \to [n]$, the graph $G_\pi = (V, E_\pi)$ defined by $(v_{\pi(i)}, v_{\pi(j)}) \in E_\pi \iff (v_i, v_j) \in E$ satisfies $g(G_\pi) = g(G)$. 
\end{dfn}

\begin{dfn}
  \label{dfn:approx}
  Let $g : \calG_n \to \R$ be a graph parameter, $\mA$ an algorithm, and $\e > 0$. We say that say that $\mA$ \dft{computes a (multiplicative) {$\bm{(1\pm\e)}$}-approximation} of $g$ if for all $G \in \calG$, the output of $\mA$ satisfies $\Pr_r(\abs{\mA(G) - g(G)} \leq \e g(G)) \geq 2/3$. Here, the probability is taken over the random choices of the algorithm $\mA$ (i.e., over the random string $\rho$). 
\end{dfn}

\begin{rem}
  In the general graph model, every graph $G$ can be explored using $O(\max\set{n, m})$ queries, for example, by using depth first search. Thus, we are interested in algorithms that perform $o(\max\set{n, m})$---or even better, $o(n)$---queries.
\end{rem}

\subsection{Communication Complexity Background}
\label{sec:cc}

In this section, we briefly review some background on two party communication complexity and state a fundamental lower bound for the disjointness function. We refer the reader to~\cite{Kushilevitz2006} for a detailed introduction. 


We consider two party communication complexity in the following setting. Let $f : \set{0, 1}^N \times \set{0, 1}^N \to \set{0, 1}$ be a Boolean function. Suppose two parties, traditionally referred to as Alice and Bob, hold $x$ and $y$, respectively, in $\set{0, 1}^N$. The \emph{(randomized) communication complexity}\footnote{Throughout this paper, all algorithms and protocols are assumed to be randomized.} of $f$ is the minimum number of bits that Alice and Bob must exchange in order for both of them to learn the value $f(x, y)$.

More formally, let $\Pi$ be a communication protocol between Alice and Bob. We assume that $\Pi$ is randomized, and that Alice and Bob have access to a shared random string, $\rho$. We say that $\Pi$ \dft{computes} $f$ if for all $x, y \in \set{0, 1}^N$, $\Pr_{\rho} [\Pi(x, y) = f(x, y)] \geq 2/3$, where the probability is taken over all random strings $\rho$. For fixed inputs $x, y \in \set{0, 1}^N$ and random string $\rho$, we denote the number of bits exchanged by Alice and Bob using $\Pi$ on input $(x, y)$ and randomness $\rho$ by $\abs{\Pi(x, y; \rho)}$. The \dft{(expected) communication cost}\footnote{It is more common in the literature to define the cost in terms of the worst case random string $\rho$ rather than expected. However, for our purposes (since we consider the expected query complexity) it will be more convenient to use expected cost. We allow our protocols to err with (small) constant probability, so this difference only affects the communication complexity by a constant factor.} of $\Pi$ is defined by 
\[
\cost(\Pi) = \sup_{x, y} \E_\rho(\abs{\Pi(x, y; \rho)}).
\]
Finally, the \dft{(randomized) communication complexity} of $f$, denoted $R(f)$, is the minimum cost of any protocol that computes $f$:
\[
R(f) = \min \set{\cost(\Pi) \sucht \Pi \text{ computes } f}.
\]

The notion of communication complexity extends to partial functions in a natural way. That is, we may restrict attention to particular inputs for $f$ and allow $\Pi$ to have arbitrary output for all other values. Formally, we model this extension to partial functions via \dft{promises} on the input of $f$. Let $P \subseteq \set{0, 1}^N \times \set{0, 1}^N$. We say that a protocol $\Pi$ computes $f$ for the promise $P$ if for all $(x, y) \in P$, $\Pr_\rho(\Pi(x, y) = f(x, y)) \geq 2/3$. The communication complexity of a promise problem (or equivalently, a partial function) is defined analogously to the paragraph above.
  
One of the fundamental results in communication complexity is a linear lower bound for the communication complexity of the disjointness function. Suppose Alice and Bob hold subsets $A, B \subseteq [N]$, respectively. The disjointness function takes on the value $1$ if $A \cap B = \varnothing$, and $0$ otherwise. By associating $A$ and $B$ with their characteristic vectors in $\set{0, 1}^N$ (i.e. $x_i = 1 \iff i \in A$ and $y_j = 1 \iff j \in B$), we can define the disjointness function as follows.

\begin{dfn}
  \label{dfn:disj}
  For any $x, y \in \set{0, 1}^N$, the \dft{disjointness function} is defined by the formula
  \[
  \disj(x, y) = \neg \bigvee_{i = 1}^N x_i \wedge y_i.
  \]
\end{dfn}

The following lower bound for the communication complexity of $\disj$ was initially proved by Kalyansundaram and Schintger~\cite{Kalyanasundaram1992} and independently by Razborov~\cite{Razborov1992}. All of the results presented in this paper rely upon this fundamental lower bound.

\begin{thm}[\cite{Kalyanasundaram1992, Razborov1992}]
  \label{thm:disj-lb}
  The randomized communication complexity of the disjointness function is $R(\disj) = \Omega(N)$. This result holds even if $x$ and $y$ are promised to satisfy $\sum_{i = 1}^N x_i y_i \in \set{0, 1}$---that is, Alice's and Bob's inputs are either disjoint or intersect on a single point.
\end{thm}

The promise in Theorem~\ref{thm:disj-lb} is known as \dft{unique intersection}. We will also use a variant of the unique intersection problem that we refer to as the \emph{\kinter{} problem}.

\begin{dfn}
  \label{dfn:k-intersection}
  Let $x, y \in \set{0, 1}^N$. We say that $x$ and $y$ are \dft{{$\bm{k}$}-intersecting} if $\sum_{i = 1}^N x_i y_i \geq k$. The \kinter{} function is defined by the formula
  \[
  \inter_k(x, y) =
  \begin{cases}
    1 & \text{if } \sum_i x_i y_i \geq k\\
    0 & \text{otherwise.}
  \end{cases}
  \]
\end{dfn}

We now prove the following consequence of Theorem~\ref{thm:disj-lb}.

\begin{cor}
  \label{cor:k-intersection}
  $R(\inter_k) = \Omega(N / k)$. The result holds even if $x$ and $y$ are promised to satisfy $\sum_i x_i y_i \in \set{0, k}$.
\end{cor}
\begin{proof}
  The argument is by simulation. Specifically, we will show that any efficient protocol for $\inter_k$ yields an efficient protocol for $\disj$. Suppose $\Pi$ is a protocol for the promise problem of the corollary with $\cost(\Pi) = B$. For $x, y \in \set{0, 1}^{N / k}$, let $x^k, y^k \in \set{0, 1}^N$ denote the concatenation of $x$ and $y$ (respectively) repeated $k$ times. Observe that if $x$ and $y$ satisfy the unique intersection promise, then $x^k$ and $y^k$ satisfy the \kinter{} promise. Further, $\inter_k(x^k, y^k) = 0$ if and only if $\disj(x, y) = 1$. Since $\Pi$ computes $\inter_k$ for all $x', y' \in \set{0, 1}^N$ satisfying the \kinter{} promise, $\Pi(x^k, y^k)$ computes $\neg \disj$ on input $x, y$. Therefore, by Theorem~\ref{thm:disj-lb}, $\cost(\Pi) = \Omega(N / k)$, which gives the desired result.
\end{proof}


\section{General Lower Bounds}
\label{sec:general-lb}
In this section, we describe a framework for obtaining general query lower bounds from communication complexity. Let $\calG_n$ denote the family of graphs on the vertex set $V = [n]$, which we assume have labels $1$ through $n$. We will use $g : \calG_n \to \set{0, 1}$ to denote a Boolean function on $\calG_n$. 


\begin{dfn}
  \label{dfn:embedding}
  Let $P \subseteq \set{0, 1}^N \times \set{0, 1}^N$. Suppose $f : P \to \set{0, 1}$ is an arbitrary (partial) function, and let $g$ be a Boolean function on $\calG_n$. Let $\calE : \set{0, 1}^N \times \set{0, 1}^N \to \calG_n$. We call the pair $(\calE, g)$ an \dft{embedding} of $f$ if for all $(x, y) \in P$ we have $f(x, y) = g(\calE(x, y))$. 
\end{dfn}

For a general embedding $(\calE, g)$ of a function $f$, the edges of $\calE(x, y)$ can depend on $x$ and $y$ in an arbitrary way. In order for the embedding to yield meaningful lower bounds, however, each allowable query $q$ should be computable from $x$ and $y$ with little communication.



\begin{dfn}
  \label{dfn:query-cost}
  Let $q : \calG_n \to \set{0, 1}^*$ be a query and $(\calE, g)$ an embedding of $f$. We say that $q$ has \dft{communication cost} at most $B$ and write $\cost_{\calE}(q) \leq B$ if there exists a (zero-error) communication protocol $\Pi_q$ such that for all $(x, y) \in P$ we have $\Pi_q(x, y) = q(\calE(x, y))$ and $\abs{\Pi_q(x, y)} \leq B$.
\end{dfn}

\begin{thm}
  \label{thm:query-lb}
  Let $Q$ be a set of allowable queries, $f : P \to \set{0, 1}$, and $(\calE, g)$ an embedding of $f$. Suppose that each query $q \in Q$ has communication cost $\cost_{\calE}(q) \leq B$. Suppose $\mA$ is an algorithm that computes $g$ using $T$ queries (in expectation) from $Q$. Then the expected query complexity of $\mA$ is $T = \Omega(R(f) / B)$.
\end{thm}
\begin{proof}
  Suppose $\mA$ computes $g$ using $T$ queries in expectation. From $\mA$ we define a two party communication protocol $\Pi_f$ for $f$ as follows. Let $x$ and $y$ denote Alice and Bob's inputs, respectively, and $\rho$ their shared public randomness. Alice and Bob both invoke $\mA$, letting their shared randomness $\rho$ be the randomness of $\mA$. Whenever $\mA$ performs a query $q$ that Alice or Bob cannot answer on their own, they communicate to the other party in order to determine the outcome of the query.\footnote{Since the randomness of $\mA$ is the shared randomness of Alice and Bob they both witness the same execution of $\mA$ and agree \emph{without communication} on which query is being performed during each step. Further, since they both know the function $\calE$, they can individually determine if a query $q$ cannot be answered by the other party. In this case, they invoke $\Pi_q$.} That is, they invoke $\Pi_{q}$ in order to compute the response $a$ to query $q$. The protocol terminates when $\mA$ halts and returns an answer $\mA(G)$, at which point Alice and Bob determine their answer to $f$ according to $\mA(G)$.

  Since $\Pr_\rho(\mA(G) = g(G)) \geq 2/3$, and $g(G) = f(x, y)$ it is clear that $\Pi_f$ computes $f$. Further, since each $\Pi_q$ satisfies $\abs{\Pi_q} \leq B(q)$, we have $\cost(\Pi_f) = 2 B\cdot T$. Since $\cost(\Pi_f) \geq R(f)$, we have $T \geq R(f) / 2 B$, as desired.
\end{proof}

Given the above Theorem, we suggest the following framework for proving graph query lower bounds.
\begin{enumerate}
\item Choose a ``hard'' communication problem $f : P \to \set{0,1}$.
\item Define functions $\calE : P \to \calG_n$ and $g : \calG_n \to \set{0, 1}$ such that $(\calE, g)$ is an embedding of $f$ in the sense of Definition~\ref{dfn:embedding}.
\item For each allowable query $q\in Q$, bound $B$, the number of bits that must be exchanged in order to simulate $q$ given $\calE$.
\end{enumerate}


\section{Lower Bounds for Particular Problems}

In this section, we derive lower bounds for particular problems. In all cases, we allow $Q$ to be the family of all degree, neighbor, and pair queries. In the case of the previously known lower bounds (estimating the number of edges, cliques, $s\th$-moment and sampling an edge from an almost uniform distribution) the graph constructions are similar or identical to the lower bound constructions in the original works. Our contribution is in the simplicity of the analysis.

\subsection{Counting Subgraphs}
\label{sec:subgraph-counting}
Let $G = (V, E)$ be a graph and $H = (V_H, E_H)$ be a fixed graph with $\abs{V_H} = k$. We denote the number of instances of $H$ in $G$ by $h_H(G)$. That is, $h_H(G)$ is the number of subgraphs $G' = (V', E')$ with $V' \subseteq V$, $E' \subseteq E$, and $\abs{V'} = k$ such that $G'$ is isomorphic to $H$.

\begin{thm}
  \label{thm:subgraph-counting}
  Let $k$ be a fixed constant, $H = (V_H, E_H)$ a fixed graph on $k$ vertices ($\abs{V_H} = k$), and $G'$ any graph on $n/2$ vertices. For any $\mu \leq \binom{n/2}{k}$, any algorithm $\mA$ that distinguishes between graphs $G$ on $n$ vertices satisfying $h_H(G) = h_H(G')$ and $h_H(G) = h_H(G') + \mu$ requires $\Omega(n / \mu^{1/k})$ degree, neighbor, or pair queries in expectation.
\end{thm}
\begin{proof}
  We apply Theorem~\ref{thm:query-lb} with $f = \disj$, the disjointness function with input size $N = \Omega(n / \mu^{1/k})$. For fixed $n$ and $x, y \in \set{0, 1}^N$ we construct a graph $G = \calE(x, y)$ on $n$ vertices as follows. Take $V = \set{v_1, v_2, \ldots, v_n}$. Let $\ell$ be the smallest integer satisfying $\binom{\ell}{k} \geq \mu$ so that $\ell = O(\mu^{1/k})$. Take $N = n / 2 \ell$. We partition the first $n / 2$ vertices into $N$ sets of size $\ell$ which we denote $K_1, K_2, \ldots, K_{N}$. That is
  \[
  K_j = \set{v_{j(\ell - 1) + 1}, v_{j(\ell - 1) + 2}, \ldots, v_{j \ell}}\;.
  \]
  The set of edges within $K_j$ is determined by $x_j$ and $y_j$. If $x_j = y_j = 1$, then $K_j$ is a clique. Otherwise, $K_j$ is a set of isolated vertices. Formally,
  \begin{equation}
    \label{eqn:clique-edge}
    \text{for all } u, v \in K_j,\ (u, v) \in E \iff x_j = y_j = 1\;.
  \end{equation}
  Edges are added to the remaining $n/2$ vertices of $V$ (i.e., vertices $V_2 = \set{v_{n/2 + 1}, \ldots, v_n}$) so that the induced subgraph on $V_2$ is isomorphic to $G'$. Finally, let $g : \calG_n \to \set{0, 1}$ be the (partial) function defined by
  \[
  g(G) =
  \begin{cases}
    1 &\text{if}\ h_H(G) \leq h_H(G')\\
    0 &\text{if}\ h_H(G) \geq h_H(G') + \mu
  \end{cases}
  \]

  We claim that $(\calE, g)$ is an embedding of $\disj$. To see this, first note that if $\disj(x, y) = 1$, then the condition of Equation~(\ref{eqn:clique-edge}) is never satisfied. Thus, $G$ is isomorphic to $G'$, plus $n/2$ isolated vertices. In particular, $h_H(G) = h_H(G')$. On the other hand, if $\disj(x, y) = 0$, then there exists $j \in [N]$ with $x_j = y_j = 1$. Thus $K_j$ is a clique on $\ell$ vertices, implying that $h_H(K_j) \geq \binom{\ell}{k} \geq \mu$. Therefore, $h_H(G) \geq h_H(G') + h_H(K_j) \geq h_H(G') + \mu$, and the claim follows.

  Finally, in order to apply Theorem~\ref{thm:query-lb}, we must show that each degree, neighbor, or pair query can be simulated by Alice and Bob (who know $x$ and $y$, respectively) using few bits. Let $u, v \in V$.
  \begin{description}
  \item[degree query] Notice that if $u \notin K_1 \cup \cdots \cup K_N$, then Alice and Bob can compute $d(u)$ with no communication, as $d(u)$ does not depend on $x$ or $y$. If $u \in K_j$, then Alice and Bob can compute $d(u)$ by exchanging $x_j$ and $y_j$, which requires $2$ bits.
  \item[neighbor query] Again, if $u \notin K_1 \cup \cdots \cup K_n$, Alice and Bob can compute $\nbr_i(u)$ without communication (by specifying some ordering on the edges of $G'$ ahead of time). For $u \in K_j = \set{v_{\ell(j - 1) + 1}, \ldots v_{\ell j}}$, Alice and Bob can again compute $\nbr_i(u)$ by exchanging $x_j$ and $y_j$. To this end, if $x_j = 0$ or $y_j = 0$, then $\nbr_i(u) = \varnothing$ for all $i$. If $x_j = y_j = 1$, then we can order the neighbors of $u = v_{\ell(j - 1) + z}$ as follows: the $i\th$ neighbor of $u$ is $v_{\ell(j - 1) + z + i}$, where the sum $z + i$ is computed modulo $\ell$.
  \item[pair query] The query $\pair(u, v)$ depends only on $x$ and $y$ if $u, v \in K_j$ for some $j$. In this case, $\pair(u, v) = 1$ if and only if $x_j = y_j = 1$. Thus Alice and Bob can simulate $\pair(u, v)$ with $2$ bits of communication.
  \end{description}
  Thus, all queries can be simulated using at most $2$ bits of communication between Alice and Bob. Therefore, by Theorem~\ref{thm:query-lb}, any algorithm that computes $g$ requires $\Omega(R(\disj) / B) = \Omega(N / 2) = \Omega(n / \ell) = \Omega(n / \mu^{1/k})$ queries, as desired.
\end{proof}

We now state some consequences of Theorem~\ref{thm:subgraph-counting}. 

\begin{cor}[Theorem~3.2 in~\cite{Goldreich2008}]
  \label{cor:edge-counting}
  Suppose $\mA$ is an algorithm that gives a $(1 + \e)$ multiplicative approximation to the number of edges in a graph using neighbor, degree, and pair queries. Specifically, for any $\e > 0$, on any input graph $G = (V, E)$ with $\abs{V} = n$ and $\abs{E} = m$, $\mA$ outputs an estimate $\mhat$ satisfying $\Pr(\abs{\mhat - m} < \e) \geq 2/3$. Then the expected query complexity of $\mA$ is $\Omega(n \big{/} \sqrt{\e m})$.
\end{cor}
\begin{proof}
  Apply Theorem~\ref{thm:subgraph-counting} where $H$ is a graph consisting of two vertices connected by a single edge. Take $G'$ to be any graph on $n/2$ nodes with $m$ edges, and take $\mu = 3 \e m$. Observe that a $(1 \pm \e)$ multiplicative approximation to the number of edges in a graph distinguishes graphs with $m$ edges from those with at least $(1 + 3 \e) m$ edges for any $\e < 1/3$.
\end{proof}

\begin{cor}[cf.~\cite{Eden2015, CLIQUES}]
  \label{cor:cliques}
  Suppose $\mA$ is an algorithm that gives a $(1 + \e)$ multiplicative approximation to the number of $r$-cliques in a graph using neighbor, degree, and pair queries. Specifically, for any $\e > 0$, on any input graph $G = (V, E)$ with $\abs{V} = n$ containing $C_r$ $r$-cliques, $\mA$ outputs an estimate $\widehat{C_r}$ satisfying $\Pr(|\widehat{C_r} - C_r| < \e) \geq 2/3$. Then the expected query complexity of $\mA$ is $\Omega(n \big{/} (\e C_r)^{1/r})$. 
\end{cor}

\begin{rem}
  The lower bound of Corollary~\ref{cor:edge-counting} is tight, as a matching upper bound is given in~\cite{Goldreich2008}. In Appendix~\ref{sec:edge-counting}, we apply Theorem~\ref{thm:query-lb} to show that any $(2 - \e)$ approximation to $m$ requires $\Omega(n^2 / m)$ queries if only degree queries are allowed. This fact was observed by Feige~\cite{feige06}, who also showed that $O(n / \sqrt{m})$ degree queries are sufficient to obtain a $2$-approximation of $m$.

  The lower bound of Corollary~\ref{cor:cliques} is tight for some ranges of the parameters $n$, $m$, and $C_r$, but not for the entire range (see~\cite{Eden2015, CLIQUES}). In Section~\ref{sec:triangles}, we apply Theorem~\ref{thm:query-lb} (and Corollary~\ref{cor:cliques}) to obtain a tight lower bound for approximating the number of triangles in a graph ($C_3$) over the entire range of parameters, thereby proving the lower bound of~\cite{Eden2015}. The same methodology can be applied to prove the lower bound of~\cite{CLIQUES} for general $r$, which we present in Appendix~\ref{sec:cliques}.
\end{rem}

\subsection{Sampling Edges}
\label{sec:edge-sampling}
In this section, we prove a lower bound on the number of queries necessary to sample an edge in a graph $G = (V, E)$ from an ``almost-uniform'' distribution $D$ over $E$. The lower bound we obtain---originally proven in~\cite{Eden2017}---is tight, as a matching upper bound is proven in~\cite{Eden2017}. Here, we use ``almost uniform'' in the sense of total variational distance:

\begin{dfn}
  Let $D$ and $D'$ be probability distributions over a finite set $X$. Then the \dft{total variational distance} between $D$ and $D'$, denoted $\tvd(D, D')$, is defined by
  \[
  \tvd(D, D') = \frac 1 2 \sum_{x \in X} \abs{D(x) - D'(x)}. 
  \]
  For $\e > 0$, we say that $D$ is \dft{$\bm{\e}$-close to uniform} if $\tvd(D, U) \leq \e$ where $U$ is the uniform distribution on $X$ (i.e., $U(x) = 1 / \abs{X}$ for all $x \in X$).
\end{dfn}

\begin{thm}[cf.~\cite{Eden2017}]
  \label{thm:edge-sampling}
  Let $0 < \e < 1/3$. Suppose $\mA$ is an algorithm that for any graph $G = (V, E)$ on $n$ vertices and $m$ edges returns an edge $e \in E$ sampled from a distribution $D$ that is $\e$-close to uniform using neighbor, degree, and pair queries. Then $\mA$ requires $\Omega(n / \sqrt{m})$ queries.
\end{thm}
\begin{proof}
  We use the same embedding $\calE$ of $\disj$ described in the proof of Theorem~\ref{thm:subgraph-counting} where $G'$ is any graph on $m'$ edges, and $\ell = \sqrt{m'}$ so that $N = n / 2 \sqrt{m}$. Thus, if any $K_j$ is a clique, the induced subgraph on $K = K_1 \cup \cdots \cup K_N$ contains at least $m / 2$ edges in $G$ (where $m$ is the number of edges in $G$). We then take $g : \calG_n \to \set{0, 1}$ to be the function whose value is $0$ if and only if $K$ contains an edge. It is clear that $(\calE, g)$ is an embedding of $\disj$. Thus, by Theorem~\ref{thm:query-lb} (and the proof of Theorem~\ref{thm:subgraph-counting}), any algorithm $\mA'$ that computes $g$ requires $\Omega(N) = \Omega(n / \sqrt{m})$ queries. 

  Let $\mA$ be an algorithm as in the statement of the theorem. We will show that by invoking $\mA$ $O(1)$ times, we can compute $g$. Thus, the lower bound on the number of queries for $\mA$ follows from the lower bound on any algorithm $\mA'$ computing $g$, as above. $\mA'$ works as follows: repeat $\mA$ $7$ times to get edge samples $e_1, \ldots, e_7$. If at least one $e_i$ satisfies $e_i \in K \times K$, return $0$, otherwise return $1$. We claim that this procedure computes $g$ (on the range of $\calE$). To see this, suppose $g(G) = 0$, i.e., at least one of the $K_i$ is a clique so that $K \times K$ contains at least $m / 2$ edges. Thus the fraction of edges in $K \times K$ is at least $1/2$, so each invocation of $\mA$ must return an edge $e \in K \times K$ with probability at least $1/2 - 1/3 = 1/6$. Therefore, if $g(G) = 0$, the probability that algorithm $\mA'$ returns $1$  (i.e., that no edge $e \in K \times K$ is sampled) is at most $(1 - 1/6)^7 < 1/3$. On the other hand, if $g(G) = 1$, then the procedure will always return $1$, as $G$ contains no edges in $K \times K$.
\end{proof}



\subsection{Counting Triangles}
\label{sec:triangles}
In this section, we prove lower lower bounds for approximately counting the number of triangles, $C_3$ in a graph. When combined with Corollary~\ref{cor:cliques} (with $r = 3$), the main result in this section gives tight lower bounds for all ranges of the parameters $n$, $m$, and $C_3$. The lower bounds (and matching upper bounds) were originally described in~\cite{Eden2015}.

\begin{thm}[cf.~\cite{Eden2015}]
  \label{thm:triangles}
  Let $G$ be a graph with $n$ vertices, $m$ edges, and $C_3$ triangles. Then any algorithm $\mA$ that computes a multiplicative approximation for $C_3$ must perform $\Omega\left(\min\left\{m, \frac{{m}^{3/2}}{C_3}\right\}\right)$ degree, neighbor, or pair queries. This lower bound holds even if $\mA$ is allowed to perform random edge queries (i.e., $\mA$ may sample a random edge in $G$ from a uniform distribution).
\end{thm}

We prove Theorem~\ref{thm:triangles} by applying Theorem~\ref{thm:query-lb} with $f = \inter_k$, where $N$ is the size of the instance of $\inter_k$. For any choice of the parameters $n$, $m$, and $C_3$, we construct an embedding $(\calE, g)$ of $\inter_k$ such that if $\inter_k(x, y) = 1$, then  $\calE(x, y)$ has (roughly) $C_3$ triangles; if $\inter_k(x, y) = 0$, then $\calE(x, y)$ is triangle-free.

\begin{proof}
  Let $n, m$, and $C_3$ be given. In order to simplify our presentation, we assume that $C_3 \geq \frac 1 2 \sqrt{m}$.\footnote{At the end of the proof we will discuss how to avoid this unnecessary assumption.} Let $\ell$ be a parameter (to be chosen later), and $N = \ell^2$ the size of an instance of $\inter_k$. We identify the set $\set{0, 1}^N$ with $\set{0,1}^{\ell \times \ell}$, so that elements $x \in \set{0, 1}^N$ are indexed by two parameters $x = (x_{ij})$ with $1 \leq i, j \leq \ell$. 

  For $x, y \in \set{0, 1}^N$, we define $G = (V, E) = \calE(x, y)$ as follows. We partition the vertex set $V$ into $5$ sets $A, A', B, B', S$ each of size $\ell$, along with an auxiliary set $C$ of size $n - 5\ell$. The set $C$ plays no role in our construction except to control the number of vertices in $G$. We denote $A = \set{a_1, a_2, \ldots, a_\ell}$, $A' = \set{a_1', a_2', \ldots, a_\ell'}$, and similarly for the remaining sets in the partition. The edge set $E$ is constructed as follows:
  \begin{itemize}
  \item For all $a \in A$, $b \in B$ and $s \in S$, we have $(a, s), (b, s) \in E$.
  \item For all $i, j \in [\ell]$ we have
    \[
    \begin{cases}
      (a_i, b_j), (a_j', b_i') \in E &\text{if } x_{i j} = y_{i j} = 1\\
      (a_i, a_j'), (b_j, b_i') \in E &\text{otherwise}.
    \end{cases}
    \]
  \end{itemize}
  See Figure~\ref{fig:triangles} for an illustration of the construction of $\calE(x, y)$.

  \begin{figure}[tb!]
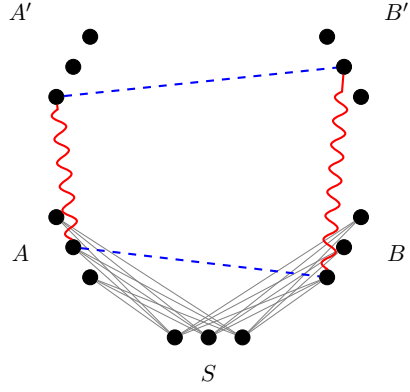

    \triangles
    \caption{An illustration of the two graph constructions of $\calE(x,y)$ for  $C_3 \geq \frac 1 2 \sqrt m$. If $x_{ij} = y_{ij} = 1$, then the dashed blue edges $(a_i,b_j), (a'_j, b'_i)$ are in $E$, and otherwise the red curly edges $(a_i,a'_j), (b_j,b'_i)$ are in $E$. Observe that each edge of the form $(a_i, b_j)$ is contained in $\abs{S}$ triangles.}
    \label{fig:triangles}
  \end{figure}

  Define the partial function $g : \calG_n \to \set{0, 1}$ by
  \[
  g(G) =
  \begin{cases}
    0 &\text{if } C_3(G) = 0\\
    1 &\text{if } C_3(G) \geq k \ell.
  \end{cases}
  \]
  With this definition of $g$, we claim that $(\calE, g)$ is an embedding of $\inter_k$. To see the claim is true, first consider the case where $\inter_k(x, y) = 0$---i.e., $x$ and $y$ are disjoint. In this case, all edges in $\calE(x, y)$ are between $A$ and $S$, $A$ and $A'$, $B$ and $S$, or $B$ and $B'$. Therefore there are no edges between $S \cup A \cup B$ and $A' \cup B' \cup C$, implying that $G$ is bipartite, hence triangle-free. On the other hand, if $\inter_k(x, y) = 1$, then for each of the (at least $k$) pairs $(i, j) \in [\ell]^2$ satisfying $x_{ij} = y_{ij} = 1$, we have $(a_i, b_j) \in E$. Therefore, for each $s \in S$, the edges $(a_i, b_j), (b_j, s), (s, a_i) \in E$ form a triangle. Since $\abs{S} = \ell$, this implies that $\calE(x, y)$ contains at least $k \ell$ triangles.

  To apply Theorem~\ref{thm:query-lb}, we must show that each allowable query can be simulated by Alice and Bob holding $x$ and $y$, respectively. For any $u, v \in V$, the queries can be simulated as follows.
  \begin{description}
  \item[degree query] $d(u)$ is independent of $x$ and $y$: if $u \in S \cup A \cup B$, then $d(u) = 2 \ell$; if $u \in A' \cup B'$, then $d(u) = \ell$; if $u \in C$, then $d(u) = 0$. Thus Alice and Bob can simulate any degree query without communication.
  \item[neighbor query] For $a_i \in A$, we label $a_i$'s incident edges so that $a_i$'s $j\th$ neighbor is either $b_j$ (if $x_{ij} = y_{ij} = 1$) or $a_j'$ otherwise. Edges incident to $A'$, $B$, and $B'$ are labeled similarly. Thus, Alice and Bob can answer queries of the form $\nbr_i(u)$ for $v \in A \cup A' \cup B \cup B'$ with $j \leq \ell$ by exchanging $x_{ij}$ and $y_{ij}$ using $2$ bits of communication. All other neighbor queries can be answered without communication.
  \item[pair query] Alice and Bob can answer pair queries of the form $\pair(a_i, a_j')$, $\pair(a_i, b_j)$, and $\pair(b_j, b_i')$ by exchanging $x_{ij}$ and $y_{ij}$. All other queries can be answered without communication. Again the communication cost is $2$ bits.
  \item[uniform edge sample] Alice and Bob can sample a random edge from a uniform distribution using their shared public randomness and the fact that each node in $\calE(x, y)$ has the same degree independent of $x$ and $y$. To this end, Alice and Bob sample $e = (u, v)$ by first sampling a vertex $v \in V$ where each node is chosen with probability proportional to its degree, $d(v)$. Alice and Bob then choose a random number $i \in [d(v)]$ uniformly at random, and sample the edge $e = (v, u)$ where $u = \nbr_i(v)$ at a communication cost of (at most) $2$ bits. Note that $e = (v, u)$ is sampled with probability
    \[
    \Pr(e = (v, u) \text{ is sampled}) = \frac{d(v)}{\sum_{w \in V} d(w)} \cdot \frac{1}{d(v)} + \frac{d(u)}{\sum_{w \in V} d(w)} \frac{1}{d(u)} = \frac{1}{m},
    \]
    so that edges are indeed sampled according to a uniform distribution.
  \end{description}

  Since all queries can be simulated using at most $2$ bits of communication between Alice and Bob and $N = \ell^2$, Theorem~\ref{thm:query-lb} (together with communication lower bound for $\inter_k$, Corollary~\ref{cor:k-intersection}) implies that computing $g$ requires $\Omega(\ell^2 / k)$ queries. For a $k$-intersecting instance (i.e., $\inter_k(x, y) = 1$), we have $m = 4 \ell^2$ and $C_3(G) \geq k \ell$. Thus, setting $\ell = \frac{1}{2} \sqrt{m}$ and $k = C_3 / \ell$, we obtain the desired result when $2 C_3 / \sqrt{m} \geq 1$. In the case where $C_3 < \frac 1 2 \sqrt m$, we may take $k = 1$ and modify the construction so that $\abs{S} = C_3$.
\end{proof}

\begin{rem}
  \label{rem:triangle-random-edge}
  Combined with Corollary~\ref{cor:cliques} with $r = 3$, Theorem~\ref{thm:triangles} gives a lower bound of $\Omega\left(\min\left\{m, m^{3/2} / C_3\right\} + n / C_3^{1/3}\right)$ queries for any algorithm that obtains a multiplicative approximation to $C_3$ in the general graph model with no random edge samples. This lower bound is tight by the matching upper bound in~\cite{Eden2015}. When random edge samples are allowed, however, the lower bound construction of Theorem~\ref{thm:subgraph-counting} (hence Corollary~\ref{cor:cliques}) does not yield the same lower bound. To see this, note that a random edge sample will be in $K_j$ with probability $\e$, so only $O(1/\e)$ such samples are sufficient distinguish $\disj(x, y) = 1$ from $0$-instances with constant probability. When random edge samples are allowed, we conjecture that the lower bound of Theorem~\ref{thm:triangles} is tight for the entire range of $n$, $m$, and $C_3$.
\end{rem}

\subsection{Computing Edge Connectivity}
\label{sec:connectivity}
In this section, we consider the problem of estimating the edge connectivity of a graph. Recall that a graph $G = (V, E)$ is \emph{$k$-(edge)-connected} if at least $k$ edges must be removed from $G$ in order to disconnect it. Equivalently, $G$ is $k$-connected if for every $u, v \in V$, there are at least $k$ edge-disjoint paths between $u$ and $v$. We prove the following lower bound for determining the connectivity of a graph $G$.

\begin{thm}
  \label{thm:connectivity}
  For $k \geq 1$, let $G$ be a graph with $n$ vertices and $m \geq 2 k n$ edges. Then any algorithm $\mA$ that distinguishes between the case where $G$ is $k$-connected and $G$ is disconnected requires $\Omega(m / k)$ degree, neighbor, or pair queries. This lower bound holds even if $\mA$ is allowed to perform random edge queries.
\end{thm}

\begin{proof}
  The proof uses a similar construction and analysis to that of Theorem~\ref{thm:triangles}. Again, we describe an embedding $(\calE, g)$ of $\inter_k$, for which we provide full details. The correctness of the embedding and simulation arguments are omitted, as they are nearly identical to those in the proof of Theorem~\ref{thm:triangles}.

  Let $\ell \geq 2k$ be parameter to be chosen later and $N = \ell^2$. Again we identify $\set{0, 1}^N = \set{(x_{ij}) \sucht 1 \leq i, j \leq \ell}$. For $x, y \in \set{0, 1}^N$, the graph $\calE(x, y) = (V, E)$ is constructed as follows. We partition $V$ into $5$ sets, $V = A \cup A' \cup B \cup B' \cup C$ where $\abs{A} = \abs{A'} = \abs{B} = \abs{B'} = \ell$, and $\abs{C} = n - 4 \ell$. Each $v \in C$ is connected to $k$ distinct vertices in $A$ arbitrarily so that $d(v) = k$. We then construct edges between $A$, $A'$, $B$, and $B'$ according to the following rule:
  \begin{equation}
    \label{eqn:connectivity}
    \begin{cases}
      (a_i, b_j'), (b_i, a_j') \in E &\text{if } x_{ij} = y_{ij} = 1\\
      (a_i, a_j'), (b_i, b_j') \in E &\text{otherwise.}
    \end{cases}
  \end{equation}
  We define the partial function $g : \calG_n \to \set{0, 1}$ by $g(G) = 1$ if $G$ is $k$-connected, and $0$ if $G$ is disconnected.

  We claim that $(\calE, g)$ is an embedding of $\inter_k$ (where we assume the promise that $\sum_{i,j} x_{ij} y_{ij} \in \set{0, k}$). In the case where $\inter_k(x, y) = 0$, there are no edges between $A \cup A' \cup C$ and $B \cup B'$, hence $\calE(x, y)$ is disconnected. If $\inter_k(x, y) = 1$, we must show that $\calE(x, y)$ is $k$-connected. We will show that there are at least $k$ edge disjoint paths between any pair of vertices in $\calE(x, y)$. We consider the following cases separately. 
  \begin{description}
  \item[Case 1:] $u, v \in A$ (or symmetrically, $u, v \in A', B$, or $B'$). From the definition of $\calE$ (Equation~(\ref{eqn:connectivity})) and the promise for $\inter_k$, there are at most $k$ pairs $(a_i, a_j') \in A \times A'$ that are \emph{not} contained in $E$. Since $\ell \geq 2 k$, this implies that $u, v \in A$ have at least $\ell - k \geq k$ common neighbors in $A'$. In particular, there are at least this many edge disjoint paths (of length $2$) between $u$ and $v$.

  \item[Case 2:] $u \in A$, $v \in A'$ (or symmetrically, $u \in B$, $v \in B'$). As before, $v$ has at least $k$ distinct neighbors, $u_1, \ldots, u_k \in A$. Further, by the analysis in Case~1, each $u_i$ has at least $k$ common neighbors with $u$. Therefore there exists a matching $(u_1, v_1), \ldots, (u_k, v_k) \in E$.\footnote{Such a matching can be found by greedily choosing common neighbors of $u$ and $u_1$, $u_2$, etc.} The paths $(u, v_i), (v_i, u_i), (u_i, v)$ for $i = 1, \ldots, k$ are then edge disjoint. (Note that it may be the case that $u_i = u$, in which case we take $(u, v)$ to be one of the matching edges and take this edge to be the corresponding path between $u$ and $v$.)

  \item[Case 3:] $u \in A$, $v \in B'$ (or symmetrically $u \in A'$, $v \in B$). Let $(u_1, v_1), \ldots, (u_k, v_k) \in A' \times B$ be the $k$ edges between $A'$ and $B$. Take $U = \set{u_1, \ldots, u_k}$ (respectively, $V = \set{v_1, \ldots, v_k}$) to be the multiset of endpoints of the edges between $A'$ and $B$ in $A'$ (respectively $B$). It suffices to show that are edge-disjoint paths from $u$ to each $u_i$ (with multiplicity)---the analogous result for $v$ and $v_i$ is identical. Let $u_1', u_2', \ldots, u_k'$ be distinct neighbors of $u$. Then, as in Case~1, each $u_i$ and $u_i'$ have at least $k$ common neighbors. By choosing one such common neighbor, $u_i''$ for each $i$ (greedily) we can form $k$ edge disjoint paths $(u, u_i'), (u_i', u_i''), (u_i'', u_i)$.

  \item[Case 4:] $u \in A$, $v \in B$ (or symmetrically, $u \in A'$, $v \in B'$). Let $(u_1, v_1), \ldots, (u_k, v_k)$ be as in Case~3. As in Case~3, there are $k$ edge-disjoint paths in $A \cup A'$ from $u$ to the $u_i$. Further, there are $k$ edge-disjoint paths from $v$ to the $v_j$ as in Case~2.

  \item[Case 5:] $u \in C$. Let $u_1, \ldots, u_k$ be the neighbors of $u$ in $A$. It suffices to show that there are edge disjoint paths from each $u_i$ to $v$. The cases $v \in A, A', B, B'$ are analogous to arguments in Cases~1--4. If $v \in C$, let $v_1, \ldots, v_k$ be the neighbors of $v$ in $A$. Since each pair $u_i, v_i$ share $k$ common neighbors in $A'$, we can assign a unique neighbor $w_i$ to each such pair so that $(u, u_i), (u_i, w_i), (w_i, v_i), (v_i, v)$ for $i = 1, \ldots, k$ are edge disjoint paths.
  \end{description}

  As in the proof of Theorem~\ref{thm:triangles}, every degree, neighbor, pair, or random edge query can be simulated by Alice and Bob using at most $2$ bits of communication. Therefore, by Theorem~\ref{thm:query-lb} and Corollary~\ref{cor:k-intersection}, $\mA$ requires $\Omega(N / k) = \Omega(\ell^2 / k)$ queries. The construction above satisfies $m = 2 \ell^2 + k(n - 4 \ell)$ so taking $\ell = k + \sqrt{k^2 + (m - k n)/2} = \Theta(\sqrt m)$ gives the desired result.
\end{proof}

\section{Discussion}
\label{sec:discussion}
In this paper, we presented a new technique for proving query lower bounds for graph parameter estimation problems. Here, we conclude with some open questions and suggestions for further work.

\paragraph{The power of random edge queries} In~\cite{Aliak}, Aliakbarpour et al.\ consider a graph query model that allows uniform random edge samples as one of its basic queries. This model is strictly stronger than the ``general graph'' query model that allows only degree, neighbor, and pair queries:~\cite{Aliak} provides an upper bound for counting the number of star sugraphs in a graph that beats the lower bound described in~\cite{GRS11} for the general graph model. In Theorems~\ref{thm:triangles} and~\ref{thm:connectivity}, our lower bounds apply to both the general graph model, as well as the stronger model with uniform random edge samples. In the constructions described in the proofs of Theorems~\ref{thm:triangles} and~\ref{thm:connectivity}, edge samples do not afford more computational power, essentially because the degree sequence of the constructions is fixed. In particular, each random edge sample can be simulated by sampling a vertex with probability proportional to its (known) degree, and using a neighbor query to sample a random incident edge. Indeed, sampling edges from a uniform distribution is equivalent to sampling vertices with probability proportional to their degrees.

The construction used for lower bound of Theorem~\ref{thm:subgraph-counting} cannot give a lower bound better than $\Omega(\min\set{m / \mu^{2 / k}, n / \mu^{1/k}})$, as a clique $K_j$ in that construction contains $\mu^{2 / k}$ edges. In particular, the lower bound for estimating $m$ (Corollary~\ref{cor:edge-counting}) becomes only $\Omega(1 / \e)$ if random edge samples are used. An upper bound of $O(n^{1/3})$ for estimating $m$ with random edge (and vertex) samples is implied by the algorithm of Motwani et al.~\cite{Motwani2007}. The authors also prove a lower bound of $\Omega(n^{1/3})$, although the construction only holds for $m = O(n^{2/3})$. We believe it is an interesting problem to characterize the complexity of estimating $m$ in the general graph model with random edge samples the over the full range of $m$.

In general, we would like to better understand the power uniform random edge samples in the general graph model. We conjecture that the lower bound of Theorems~\ref{thm:triangles} (and more generally, Theorem~\ref{thm:cliques_second}) is tight over the entire range of the parameters. The algorithm of Eden et al.~\cite{Eden2015} proves that the lower bound is tight for $\min\set{m, m^{3/2}/C_3} = \Omega(n / C_3^{1/3})$ even without edge samples. Thus, edge samples may only help in the regime where $n / C_3^{1/3} = \omega(\min\set{m, m^{3/2}/C_3})$.

\begin{que}
  For what graph parameter estimation problems do random edge samples help? 
\end{que}


\paragraph{Property testing lower bounds} The graph query access models we consider were initially proposed in the context of property testing~\cite{Goldreich1998, GR02, Parnas2002, KKR04}. In graph property testing with the general graph model~\cite{KKR04}, the goal is to distinguish graphs that satisfy some property $P$ from those that are far from satisfying $P$ in the sense that an $\e$-fraction of edges of the graph must be modified in order to the graph to satisfy $P$. In this model, our constructions imply property testing lower bounds, at least for some range of $m$. For example, we state a consequence of Theorem~\ref{thm:subgraph-counting} for testing the property of triangle-freeness (i.e., that $C_3(G) = 0$).

\begin{cor}
  \label{cor:triangle-pt}
  Any property testing algorithm for triangle-freeness in the general graph model requires $\Omega(n / \sqrt{\e m})$ queries.
\end{cor}

Corollary~\ref{cor:triangle-pt} follows from the construction in Theorem~\ref{thm:subgraph-counting} by taking $G'$ to be any triangle-free graph on $n/2$ vertices and $m$ edges, and $\ell = \sqrt{\e m}$. This way each (potential) clique $K_j$ on $\ell$ vertices will contain roughly $\e m$ edges. Further, by Tur\'{a}n's theorem~\cite{Turan1941, Aigner2010}, at least (roughly) $\e/3 m$ edges must be removed from $G$ (in particular from $K_j$) in order to make $G$ triangle-free in the case where $K_j$ is a clique. The lower bound of Corollary~\ref{cor:triangle-pt} matches the known lower bound due to Alon et al.~\cite{Alon2008} in the regime where the average degree $d = 2 m / n$ satisfies $d = O(n^{1/3})$. In the range $d = \omega(n^{1/3})$, the lower bounds in~\cite{Alon2008} are strictly stronger.\footnote{Alon et al.~\cite{Alon2008} prove that a lower bound of $\Omega(n^{1/3})$ holds even for all $d = O(n^{1 - \nu(n)})$ for some function $\nu(n) = o(1)$. This is in contrast for the case $d = \Omega(n)$, where $O(f(\e))$ queries are sufficient for some function $f$.} 

In the dense graph model~\cite{Goldreich1998}, only pair queries are allowed, but the distances between graphs are normalized by $n^2$ (rather than $\abs{E}$ as in the general graph model). Thus every graph with $m = o(n^2)$ is $\e$-close to the graph with no edges. In this case, the types of embeddings we present in this paper cannot yield lower bounds that are better than $\Omega(1/\e)$. Specifically, in all embeddings of disjointness we consider, each edge in $\calE(x, y)$ depends on a single bit of $x, y \in \set{0, 1}^N$. However, the value of $\disj(x, y)$ can vary by changing a single bit of $x$ or $y$. In order for Theorem~\ref{thm:query-lb} to give property testing lower bounds via an embedding of $\disj$, changing a single bit of $x$ or $y$ must change $\e n^2$ edges in $\calE(x, y)$. Thus, to obtain stronger lower bounds (e.g., lower bounds that grow as a function of $n$), either different communication primitives must be considered, or the embedding $\calE$ must be more complicated (using some nontrivial encoding of $x$ and $y$).

\begin{que}
  Can Theorem~\ref{thm:query-lb} be applied to obtain any nontrivial (i.e., $\omega(1/\e)$) lower bound for any ``natural'' graph problem in the dense graph property testing model?
\end{que}

\bibliographystyle{plain}
\bibliography{lb-CC}

\clearpage
\appendix

\section{Counting Edges with Degree Queries}
\label{sec:edge-counting}
In this appendix, we prove the following result due to Feige~\cite{feige06}.

\begin{thm}[\cite{feige06}]
  \label{thm:edge-counting}
  For any $\e > 0$, any algorithm $\mA$ that for any graph $G$ with $m = \Omega(n)$ computes an estimate $\widehat{m}$ of $m$ that satisfies $m \leq \widehat{m} \leq (2 - \e) m$ using only degree queries requires $\Omega(n^2 / m)$ queries.
\end{thm}
\begin{proof}
  We apply Theorem~\ref{thm:query-lb} using $f = \disj$ with $N = n / 3 k$ (where $k$ is a parameter to be chosen later) and the promise that $x, y \in \set{0, 1}^N$ are either disjoint or uniquely intersecting. We construct $\calE(x, y) = (U \cup V \cup W, E)$ as follows. We partition the vertices of $\calE$ into three sets, $U$, $V$, and $W$ each of size $n / 3$. For any $k \leq n/3$, we partition $U = U_1 \cup U_2 \cup \cdots \cup U_k$ into $N = n / 3 k$ sets of size $k$, and similarly with $V$ and $W$.


  The edge set $E$ of $\calE(x, y)$ is constructed as follows. If $\disj(x, y) = 1$, then $U$ is a set of isolated vertices, and for each $i \in [N]$, $V_i \cup W_i$ is a complete bipartite graph. On the other hand, if $\disj(x, y) = 0$ with $x_j = y_j = 1$, then each vertex in $V \cup W$ shares an edge with each $u \in U_j$ and there are no other edges. 
  Formally,
  \[
  E =
  \begin{cases}
    \bigcup_{i = 1}^N \set{(v, w) \in V_i \times W_i} &\text{if } \disj(x, y) = 1\\
    \paren{U_j \times V} \cup \paren{U_j \times W} &\text{if } x_j = y_j = 1.
  \end{cases}
  \]

  Observe that if $\disj(x, y) = 1$, then $\calE$ satisfies $m = n k / 3$, while if $\disj(x, y) = 0$, then $m = 2 n k / 3$. We define the partial function $g : \calG_n \to \set{0, 1}$ by
  \[
  g(G) =
  \begin{cases}
    1 &\text{if } m \leq n k / 3\\
    0 & \text{if } m \geq 2 n k / 3.
  \end{cases}
  \]
  It is then clear that $(\calE, g)$ is an embedding of $\disj$.

  In order to apply Theorem~\ref{thm:query-lb}, we must show that degree queries can be efficiently simulated by Alice and Bob holding $x$ and $y$, respectively. To this end, for all $v \in V \cup W$, we have $d(v) = k$, independent of $x$ and $y$. Hence, such a query can be simulated without communication. For any $j \in [N]$ and $u \in U_j$, Alice and Bob can compute $d(u)$ by exchanging $x_j$ and $y_j$. Specifically,
  \[
  d(u) =
  \begin{cases}
    2 n / 3 &\text{if } x_j = y_j = 1,\\
    0 &\text{otherwise.}
  \end{cases}
  \]
  Since each degree query can be simulated using at most $B = 2$ bits of communication, Theorems~\ref{thm:query-lb} and~\ref{thm:disj-lb} together imply that any algorithm that distinguishes graphs satisfying $m \leq n k / 3$ from those satisfying $m \geq 2 n k / 3$ requires $\Omega(N)$ degree queries. Since $N = n / 3 k$ and $m = n k / 3$ in the disjoint case, we have $N = \Theta(n^2 / m)$, and the theorem follows.
\end{proof}


\section{Counting Cliques}
\label{sec:cliques}
\begin{thm}[\cite{CLIQUES}]
  \label{thm:cliques_second}
  Any algorithm $\mA$ that for every graph on $n$ vertices and $m$ edges computes a multiplicative approximation to $C_r$---the number of $r$-cliques in the graph---must perform $\Omega\left(\min\left\{m, \frac{{m}^{r/2}}{C_r\cdot (c r)^r}\right\}\right)$ queries, for some absolute constant $c > 0$.
\end{thm}

\begin{proof}
  As with the proof of Theorem~\ref{thm:triangles} we appeal to Theorem~\ref{thm:query-lb} using \kinter\ as the communication primitive; the construction is a generalization to that in the proof of Theorem~\ref{thm:triangles}. For simplicity, we initially assume that $C_r = \Omega(m^{r/2 -1})$---the case for smaller $C_r$ is discussed at the end of the proof. As before, $\ell$ is a parameter and $N = \ell^2$ is the size of the $\inter_k$ instance. For all $x, y \in \set{0, 1}^N = \set{0, 1}^{\ell \times \ell}$, $\calE(x, y) = (V, E)$ is constructed as follows. We partition $V$ into $r + 3$ sets, $A, A', B, B', S_1, \ldots, S_{r - 2}, C$. The first $r + 2$ sets each have size $\ell$, while $C$ contains $n - (r + 2) \ell$ isolated vertices. The sets $A, A', B, B'$, and $S = S_1 \cup \cdots \cup S_{r-2}$ play analogous roles to the corresponding sets of the proof of Theorem~\ref{thm:query-lb}.

  The edges of $\calE(x, y)$ are defined as follows. For all $i < j \leq r - 2$, $E$ contains all possible edges between $S_i$ and $S_j$. That is, for all $s_i \in S_i$ and $s_j \in S_j$ we have $(s_i, s_j) \in E$. For all $a \in A$, $b \in B$, and $s \in S = S_1 \cup \cdots \cup S_{r - 2}$, we also have $(a, s), (b, s) \in E$. Finally, for all $i, j \in [\ell]$, edges between $A, A', B$ and $B'$ are determined from $x$ and $y$:
  \[
  \begin{cases}
    (a_i, b_j), (a_j', b_i') \in E &\text{if } x_{i j} = y_{i j} = 1\\
    (a_i, a_j'), (b_j, b_i') \in E &\text{otherwise}.
  \end{cases}
  \]
  We refer to Figure~\ref{Fig:cliques} for an illustration with $r = 2$.

  \begin{figure}[tb!]
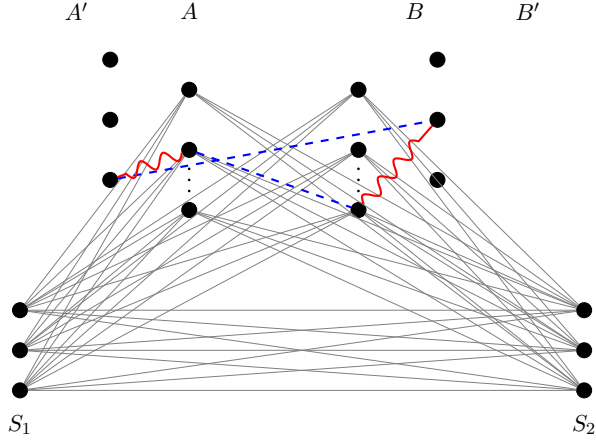

    \families
    \caption{An illustration of the two graph constructions of $\calE(x,y)$ for $r=4$ and $C_r =\Theta(k \cdot M^{(r-2)/2})$. If $x_{ij} = y_{ij} = 1$ then the dashed blue edges $(a_i,b_j), (a'_j, b'_i)$ are in $E$, and otherwise the red curly edges $(a_i,a'_j), (b_j,b'_i)$ are in $E$.}
    \label{Fig:cliques}
  \end{figure}

  Observe that in the case where $\inter_k(x, y) = 0$ (i.e, $x$ and $y$ are disjoint), then $\calE(x, y)$ does not contain any $r$-cliques---$C_r = 0$. However, if $x_{i j} = y_{i j} = 1$, then for all $(s_1, s_2, \ldots, s_{r-2}) \in S_1 \times S_2 \times \cdots \times S_{r-2}$, the set $\set{a_i, b_j, s_1, s_2, \ldots, s_{r-2}}$ is a clique in $\calE(x, y)$. Since there are $\ell^{r-2}$ such choices of $s_1, s_2, \ldots, s_{r-2}$, the edge $(a_i, b_j)$ is contained in $\ell^{r-2}$ cliques of size $r$. Thus, if $\inter_k(x, y) = 1$, we have
  \begin{equation}
    \label{eqn:clique-count}
    C_r(\calE(x, y)) \geq k \ell^{r - 2}.
  \end{equation}
  We compute the number of edges in $\calE(x, y)$ to be
  \begin{equation}
    \label{eqn:edge-count}
    m = \ell^2 \paren{\binom{r-2}{2} + 2 (r - 2) + 2} = \Theta(\ell^2 r^2).
  \end{equation}

  Every degree, neighbor, and pair query, as well as random edge samples can be simulated by Alice and Bob using at most $2$ bits of communication per query---the proof is identical to the argument in the proof of Theorem~\ref{thm:triangles}. Therefore, Theorem~\ref{thm:query-lb} and Corollary~\ref{cor:k-intersection} imply that distinguishing graphs with $C_r = 0$ from those with $C_r \geq k \ell^{r - 2}$ requires $\Omega(\ell^2 / k)$ queries. Solving Equation~\ref{eqn:clique-count} for $k$, and Equation~\ref{eqn:edge-count} for $\ell$ yield\footnote{For all graphs with $m$ edges, we have $C_r = O(m^{r/2})$. Our construction asymptotically achieves this upper bound on $C_r$ when $k = \Omega(\ell^2)$. Thus, the construction yields (asymptotically) all possible ranges of $n$, $m$, and $C_r$.}
  \[
  k = \frac{C_r}{\ell^{r-2}} \quad\text{and}\quad \ell = \Theta\paren{\frac{\sqrt m}{r}}\;.
  \]
  Substituting these expression into the lower bound $\Omega(\ell^2 / k)$, we find that $\Omega(m^{r/2} / C_r (c r)^r)$ queries are necessary to obtain any multiplicative approximation to $C_r$ for some constant $c > 0$.

  In the case where $C_r = o(m^{r/2 - 1})$, we modify the construction above as follows. We remove edges from within $S$ (i.e., edges between the various $S_i$ and $S_j$) so that $S$ contains $\Theta(C_r)$ cliques of size $r - 2$. We take $k = 1$ so that (as before) if $\inter_k(x, y) = 0$, $\calE(x, y)$ contains no $r$-cliques, while if $\inter_k(x, y) = 1$, $\calE(x, y)$ contains $\Theta(C_r)$ cliques. In particular, the single edge $(a_i, b_j)$ corresponding to $x_{ij} = y_{ij} = 1$ participates in $\Theta(C_r)$ cliques: one for each $(r-2)$-clique within $S$. Since all edges between $A, B$, and $S$ remain, we maintain that $m = \Theta(\ell^2)$.
\end{proof}


\section{Computing Moments of the Degree Distribution}
\label{sec:moments}
Let $M_s=\sum_{v\in V}d^s(v)$ denote the $s\th$ moment of the degree distribution of a graph. The problem of estimating $M_s$ in sublinear time in the general graph model was first studied in~\cite{GRS11} for general graphs and later in~\cite{Aliak} in a slightly different model which allows degree queries as well as access to uniform edge samples. In~\cite{DegMoments}, Eden et al. generalize the results of~\cite{GRS11} to graphs with bounded arboricity in the general graph model. The arboricity of a graph $G$ is a measure of its sparseness, that is essentially equal to the maximum of the average degree over all subgraphs $S$ of $G$ as proved in \cite{NW61, NW64}. In~\cite{DegMoments}, Eden et al.\ observe that the hardness of the moments estimation problem arises from the existence of small (hidden) dense subgraphs, and exploit the fact that in graphs with bounded arboricity, no such dense subgraphs can exists. This allows them to devise an algorithm for estimating the moments that has improved query complexity when the algorithm is given an upper bound on the arboricity of the graph.

We prove the result of~\cite{DegMoments} which gives a lower bound of
\[
\Omega\left(\frac{n\alpha^{1/s}}{M_s^{1/s}} + \min\left \{\frac{n \cdot \alpha}{M_s^{1/s}},\; \frac{n^s \cdot \alpha}{M_s},\;  n^{1-1/s},\; \frac{n^{s-1/s}}{M_s^{1-1/s}} \right \} \right)
\]
for the problem of estimating $M_s$ when the graph has $n$ vertices and arboricity at most $\alpha$. We note that for every graph it holds that $\alpha\leq \sqrt m$, and hence if an upper bound on the arboricity is not known, then one can substitute $\alpha=\sqrt m$ and get the lower bound of estimating $M_s$ for general graphs. For more details see~\cite{DegMoments}.

We start with the following useful claim and definition.
\begin{claim}[Claim 12 and Footnote 4 in \cite{DegMoments}]\label{clm:alpha-bounds}
	For any graph $G$ with arboricity $\alpha(G)$, 
	\[\frac{M_s(G)}{n^s} \leq \alpha(G)\leq M_s(G)^{\frac{1}{s+1}}\;. \]
\end{claim}

\begin{dfn}
        For a value $\wtM_s$, we let $\gmom:\mG_{n} \to \{0,1\}$ be defined by
	\[
	\gmom(G) =
	\begin{cases}
	0 & \text{if } M_s(G) \leq \wtM_s\\
	1 & \text{if } M_s(G) \geq c\cdot \wtM_s\,,
	\end{cases}
	\]
        where $c$ is a fixed constant to be determined later.
\end{dfn}

\begin{thm}[Thm. 7 in \cite{DegMoments}]
  \label{thm:moments_1term}
	Let $G$ be a graph over $n$ vertices and with arboricity $\alpha$, and
  let $\mA$ be a constant-factor approximation algorithm for $M_s(G)$ with allowed queries $Q$, consisting of neighbor, degree, and pair queries. The expected query complexity of $\mA$ is $\Omega\left(\frac{n\alpha^{1/s}}{M_s^{1/s}(G)}\right)$.
\end{thm}

\begin{proof}
  The proof uses an embedding of $\disj$, and is similar to the proof of Theorem~\ref{thm:subgraph-counting}. We modify the construction of the hidden subgraph so as to contribute a constant factor of $M_s$ to the resulting graph without increasing the arboricity of the graph. For a fixed $\wtM_s$, let we take $H$ to be a complete bipartite graph between sets of vertices of size $(c \cdot \wtM_s/\alpha)^{1/s}$ and $\alpha$ ($H$ will play the role of the cliques in the proof of Theorem~\ref{thm:subgraph-counting}). Then $M_s(H) = \alpha \cdot (c\cdot \wtM_s/\alpha) + (c\cdot \wtM_s/\alpha)^{1/s}\cdot \alpha^s \geq c \wtM_s$. Further, the arboricity $H$ is $\alpha$.\footnote{To see this, note that the edges of $H$ can be partitioned into $\alpha$ trees: each tree consists of the edges incident with a single vertex the $\alpha$-sized side of the bipartition.}

  The graph $G = (V, E) = \calE(x, y)$ is constructed from an arbitrary graph $G' = (V', E')$ on $n'$ vertices where $M_s(G')=\wtM_s$. We then form $V = V' \cup W$, where $\abs{W} = n'$. For fixed $k$, we partition $W$ into subsets each of size $k$, $W_1 \cup \cdots \cup W_{n' / k}$, where $k = (c \cdot \wtM_s/\alpha)^{1/s}+\alpha$. As in the proof of Theorem~\ref{thm:subgraph-counting}, each $W_i$ will be either a set of isolated vertices (if $i \notin x \cap y$) or isomorphic to $H$ (if $i \in x \cap y$). The calculation of $M_s(H)$ above implies that in the intersecting case $M_s(G) \geq (1 + c)\wtM_s$, while $M_s(G) = \wtM_s$ in the disjoint case. The remainder of the proof of Theorem~\ref{thm:moments_1term} is analogous to Theorem~\ref{thm:subgraph-counting}.
\end{proof}

We next give an alternate proof for Theorem~8 in~\cite{DegMoments}. The embedding $\calE$ we construct yields graphs that are almost identical to the constructions in the original proof. 
We note that some details in the original proof are omitted in~\cite{DegMoments} as the argument in~\cite{DegMoments} is based on the proofs of the second and third items of Theorem~5 in~\cite{GRS11}.
\sloppy
 \begin{thm}[Theorem~8 in~\cite{DegMoments}]
 	\label{thm:moments_2term}
	Let $G$ be a graph over $n$ vertices and with arboricity $\alpha$, and
        let $\mA$ be a constant-factor approximation algorithm for $M_s(G)$ with allowed queries $Q$, consisting of neighbor, degree, and pair queries. The expected query complexity of $\mA$ is \[\Omega\left(\min\left \{\frac{\alpha(G)\cdot n }{M_s(G)^{1/s}},\; n^{1-1/s},\;\frac{\alpha(G) \cdot n^s }{M_s(G)},\;  \frac{n^{s-1/s}}{M_s(G)^{1-1/s}} \right \}   \right).\]
\end{thm}
 \begin{proof}
   Once again, we reduce from the problem of set disjointness. Let $f=\disj$ for $N\in \left \{\frac{\alpha(G)\cdot n}{M_s(G)^{1/s}},\; n^{1-1/s},\; \frac{\alpha(G) \cdot n^s}{M_s(G)},\;   \frac{n^{s-1/s}}{M_s(G)^{1-1/s}} \right \}$. In order to prove the different lower bounds we describe a function $\calE:\{0,1\}^N\times \{0,1\}^N \to \mG_{n'}$ that depends on the relations between $n,\alpha$ and $M_s$, and so that $n'=\Theta(n), \alpha(G)=\alpha$ and the value of $M_s(G)$ is some constant factor of $\wtM_s$ depending on whether or not $x$ and $y$ are disjoint. 

   As in~\cite{DegMoments}, we divide the proof into two cases depending on the relation between $\wtM_s$ and $n$. In each case consider two sub-cases that depend on the relation between $\alpha$ and $(\wtM_s/n)^{1/s}$.
 
   \paragraph{The case $\bm{M_s^{1/s}\leq n/c}$ for some constant $\bm{c>4}$}
   
   We start with some intuition behind the construction of the graphs. Let $G$ be a graph $G=A\cup B \cup R \cup W$ as follows. $A \cup B$ is a $d$-regular bipartite subgraph  over $2n$ vertices (where $A=\{a_1,\ldots, a_n\}$ and $B=\{b_1,\ldots, b_n\}$), $R$ is a clique over $\alpha$ vertices and $W$ consists of $d/\ell$ subgraphs $W_1, \ldots, W_{d/\ell}$ of $c$ isolated vertices each. In the case where $i \in x \cap y$, we modify the graph above such that $W_i$ consists of extremely high degree vertices so that its contribution to the $s\th$ moment amounts to a constant factor of $M_s$. 

   Specifically, Consider the following embedding $\calE:\{0,1\}^{d/\ell} \times \{0,1\}^{d/\ell}  \to \calG_{n'}$. Given $x,y \in \{0,1\}^{d/\ell}$, $G$ is a graph over $n'=2n+c\cdot (d/\ell)+\alpha$ vertices, where the exact values of $d$ and $\ell$ depend on the sub-case at hand, and will be set later. We show that in all cases $n'=\Theta(n)$ and $\alpha(G)=\alpha$. Depending on whether  or not $x$ and $y$ intersect, $M_s(G)\leq \wtM_s$ or $M_s(G)\approx c\cdot \wtM_s$. 

   For every vertex in $A\cup B$, we think of its set of neighbors as divided into $d/\ell$ blocks of size $\ell$ each. If the sets are disjoint, then all the neighbors in every block are connected to vertices on the corresponding side of the bipartite subgraph. Thus, for $a_i\in A$ a vertex's neighbors are $b_{i}, b_{i+1 \text{(mod n)}}, \ldots, b_{i+d \text{(mod n)}}$ (and analogously for $b_i \in B$). Otherwise, if $x$ and $y$ intersect on the $j\th$ index, then the $j\th$ block of neighbors of $a_i$ is in $W_j$ (rather than in $B$). In this case, the degree of every vertex in $W_j$ is $2n \cdot \ell/c$. The identity of the neighbors of $a_i \in A\cup B$ in $W_j$ depend on the value of $i$, where the first (by order of indices) $2\ell\cdot n/c$ vertices in $A\cup B$ are connected to the first $\ell$ vertices in $W_j$, and so on. See Figures~\ref{Fig:lb_families1} and~\ref{Fig:lb_families2} for an illustration.

\begin{figure}[tb!]
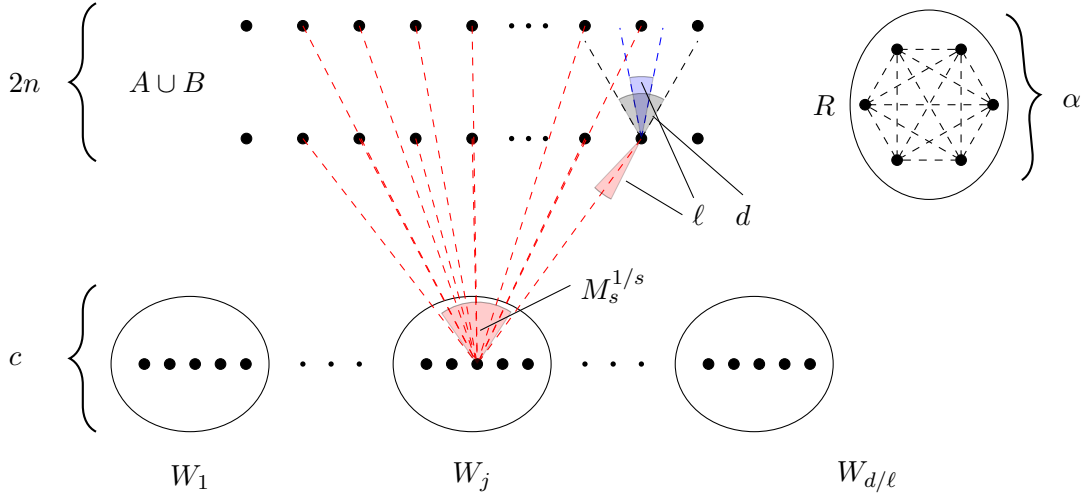
	
`	\MomentsLB
\caption{The construction of the graph $G=\calE(x,y)$. The subgraph $A\cup B$ is either a bipartite graph. The subgraph $R$ is a clique of size $\alpha$ and additionally there are $d/\ell$ subgraphs $W_1, \ldots, W_{\ell/d}$ of size $c$ each, for some small constant $c$. In the case that $x$ and $y$ are disjoint $A \cup B$ is $d$-regular bipartite subgraph and all the sets $W_i$ are independent sets. In case $x$ and $y$ intersect on the $j\th$ coordinate, $A \cup B$ is a $(d-\ell)$-regular bipartite subgraph, all vertices in $A\cup B$ have $\ell = \lceil c\cdot \wtM_s^{1/s}/2n\rceil$ neighbors in $W_j$, and all vertices in $W_j$ have $\wtM_s^{1/s}$ in $A\cup B$. All other subgraphs $W_i$ for $i\neq j$ are independent sets. Hence, in the former case $M_s(G)\leq 3\wtM_s$ and in the  latter, $M_s(G) \approx c \wtM_s$.
}
	\label{Fig:lb_families1}
\end{figure}

\begin{figure}[tb!]
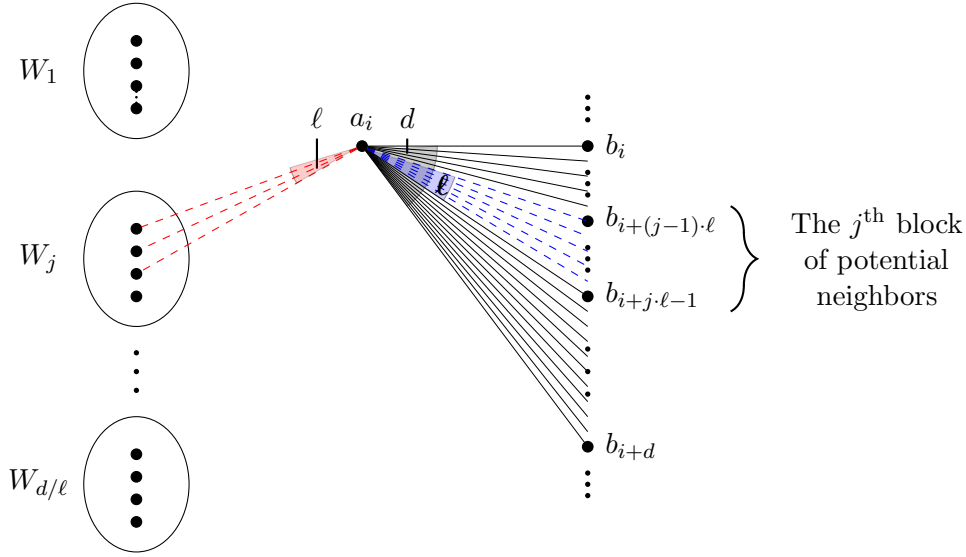
	
	\drawNbrs
	\caption{An illustration of the set of neighbors of the vertex $a_i\in A$ for the graph $\calE(x,y)$. 
	$a_i$ has $d$ neighbor, divided into $d/\ell$ blocks of size $\ell$. If $x$ and $y$ are disjoint then the identity of all the neighbors in each block is in $B$, so that for every $r\in [d]$, $a_i$'s $r\th$ neighbor is $b_{i+r-1 (\text{mod } n)}$. Otherwise, if $x$ and $y$ intersect on the $j\th$ index, then the $j\th$ block of neighbor of $a_i$ is in $W_j$.  
	The exact indices of the neighbors of $a_i$ in $W_j$ is determined according to $i$ as described in the construction. 
}
	\label{Fig:lb_families2}
\end{figure}

In the first sub-case $\alpha < (\wtM_s/n)^{1/s}$ we set $\ell=\lceil c\cdot \wtM_s^{1/s}/2n\rceil$ and $d=\alpha$. 
If $x$ and $y$ are disjoint then $M_s(G)=2n \cdot \alpha^s+\alpha^{s+1}\leq 3\wtM_s$, where the last inequality is due to Claim~\ref{clm:alpha-bounds}, and $\alpha(G)=\alpha$ due to the set $R$. Otherwise, if $x$ and $y$ intersect on index $j$, then the vertices of $W_j$ all have degree $2n \cdot \ell/c=\wtM_s^{1/s}$ and therefore $M_s(G)=c\cdot M_s^{1/s}+\alpha^{s+1}\approx c \cdot \wtM_s$ (by Claim~\ref{clm:alpha-bounds}) and $\alpha(G)=\alpha$ (here too it can be verified that $R$ is the subgraph that maximizes the average degree).

In the second sub-case, $\alpha \geq (\wtM_s/n)^{1/s}$, we set $\ell=\lceil c\cdot \wtM_s^{1/s}/2n\rceil$ as before, but now we set $d=\lfloor (\wtM_s/n)^{1/s}\rfloor$. Hence, in case $x$ and $y$ intersect we get $M_s(G) = c\cdot M_s + 2n \cdot (\lceil c\cdot \wtM_s^{1/s}/2n\rceil)^s+\alpha^{s+1}\approx c\cdot \wtM_s$, and otherwise $M_s(G)=2n \cdot (\lceil c\cdot \wtM_s^{1/s}/2n\rceil)^s+\alpha^{s+1} \leq 3\cdot \wtM_s$. Still in both cases $\alpha(G)=\alpha$. 

Therefore, in both sub-cases, $g(\calE(x,y))=\disj(x,y)$, and $(\calE,g)$ is an embedding of $f$.

It remains to prove that Alice and Bob can answer queries of $\mA$ according to $\calE(x,y)$ efficiently.
  \begin{description}
  \item[degree query] For the query $\degree(v)$, if $v \in A\cup B$ then Alice and Bob answer $d(v)=d$, and if $v$ is in $R$, then Alice and Bob answer $d(v)=\alpha$. If $v \in W_j$ for some $j \in d/\ell$ then Alice and Bob communicate to each other whether or not $j$ is in their set. If $x_j \wedge y_j$ then Alice and Bob announce that the sets intersect. Otherwise, if $j \notin x\wedge y$,  they answer $d(v)=0$.
    
  \item[neighbor query] For $q=\nbr_r(v)$, if $v=t_i \in R$ then Alice and Bob answer $\nbr_r(v)=t_{i+r (\text{mod } n)}$. If $v=a_i \in A$ then  Alice and Bob communicate to each other whether the $j\th$ index is in their set for $j=\lfloor r/\ell \rfloor$  (since the $r\th$ neighbor of $a_i$ belongs to the $(\lfloor r/\ell \rfloor)\th$ block of $a_i$'s neighbors). If $x_j \wedge y_j$ then Alice and Bob determine that $\nbr_j(v)$ is the $k\th$ vertex in $W_j$ where $k = j \mod \lfloor{r / \ell}\rfloor$. Otherwise, if $x_j \wedge y_j = 0$,  they answer $b_{i+r \text{(mod n)}}$. The case that $v=b_i \in B$ is analogous. Finally, if $v \in W_j$ then Alice and Bob communicate to each other whether the $j\th$ index is in their set and as before, if $j \in x\wedge y$ then they announce $\nbr_r(v)$ is the appropriate $a_i \in A$, and otherwise they answer $\nbr_r(v)=\varnothing.$
    
  \item[pair query] For $\pair(v_i,v_r)$, if the two vertices are in $R$ then Alice and Bob answer $(v_i,v_r)\in E$, and if only one of the vertices is in $R$ then they answer $(v_i,v_r)\notin E$. If the two vertices are in different sides of $A\cup B$ then assume without loss of generality that $v_i \in A$, $v_r \in B$ and $i<r$. If $r-i \notin [d]$ then Alice responds $\pair(v_i,v_r)=1$. Otherwise, Alice and Bob exchange $x_j$ and $y_j$ for $j=\lfloor (i-r)/\ell \rfloor$ (since $v_r$  in the $j\th$ block of $v_i$'s neighbors). If $x_j\wedge y_j = 1$ then they find $(v_i, v_r) \notin E$, and otherwise they answer $(v_i,v_r)\in E$.
\end{description}

It holds by the above, that every query made by the algorithm can be answered by $O(1)$ communication, and since $\mA$ can be used to compute $\gmom$, it follows from Theorem~\ref{thm:query-lb} that the expected query complexity of $\mA$ is $\Omega\left(d/\ell\right)=
\Omega\left(\min\left\{ \frac{n\alpha^{1/s}}{M_s^{1/s}}, \;n^{1-1/s}\right\}\right)$.


\paragraph{The case $\bm{M_s^{1/s}> n}$}

The construction of the graphs in this case is very similar to the previous case, except that now the sizes of the sets $C_1, \ldots, C_{d/\ell}$  is increased to $k=\lceil c\cdot \wtM_s/(2n)^s \rceil$ for a small constant $c$, and their potential contribution to the degree of the vertices of $A\cup B$, $\ell$ is also increased to $k$. Hence, if $x$ and $y$ intersect on the $j\th$ index, then the degree of the vertices in $W_j$ is $2n$ and the subgraph $W_j\cup (A\cup B)$ is a complete bipartite graph.  

As before, in the sub-case $\alpha < (\wtM/n)^{1/s}$, we set $d=\alpha$ and in the sub-case  $\alpha \geq (\wtM/n)^{1/s}$ we set $d=\lfloor (\wtM_s/n)^{1/s} \rfloor$. As noted in the proof of Theorem~8 in \cite{DegMoments}, we may assume without loss of generality that $\wtM_s \leq n^s\cdot \alpha/c'$ for a sufficiently large constant $c'$ since otherwise the lower bound  $\Omega\left(n^s\cdot \alpha/(\wtM_s)\right)$ becomes trivial. Therefore, indeed $d>\ell$ as required by our construction. Similarly we can assume without loss of generality that $\wtM_s\leq n^s\cdot \alpha/c'$  or else the lower bound $\Omega\left(n^s\cdot\alpha/(\wtM_s)\right)$ becomes trivial.
By the above settings, in the first sub-case, if  $x$ and $y$ are disjoint then $\alpha(G)=G$ and $M_s(G)=2n \cdot \alpha^s + \alpha^{s+1}\leq 3\wtM$, where the last inequality is by Claim~\ref{clm:alpha-bounds}. Otherwise, if $x$ and $y$ intersect on the $j\th$ index, then the vertices of $W_j$ all have degree $2n \cdot \ell/k$ and therefore $M_s(G)=k\cdot (2n \cdot \ell/k)^s + 2n \cdot \alpha^s + \alpha^{s+1}=2c\cdot \wtM_s + 2n \cdot \alpha^s+\alpha^{s+1} \approx c\wtM.$  Also by the choice of parameters above, the average degree in the subgraph $W_j \cup A\cup B$ is now $k$, but since we assumed $\wtM_s \leq n^s\cdot \alpha/c'$, it still holds that $R$ is the subgraph that maximizes the average and that $\alpha(G)=\alpha$. 

In the second sub-case $\alpha>(M_s/n)^{1/s}$,  we have that if 
$x$ and $y$ are disjoint, then $M_s(G)=2n \cdot (\wtM/n)	+\alpha^{s+1} \leq 3\wtM_s$, and if $x$ and $y$ do intersect then $M_s(G)=k\cdot (2n \cdot \ell/s)^s + 2n \cdot (\wtM_s/n)	 + \alpha^{s+1}\approx\cdot \wtM_s $. Here too, regardless of $x$ and $y$, $\alpha(G)=\alpha.$

As in the former case, it can be easily verified that Alice and Bob can answer any query made by $\mA$ from  $Q$  using  $O(1)$ communication, and it follows that the expected running time of $\mA$ is $ \Omega(d/\ell)=\Omega(d/k)=\Omega\left(\min\left\{\frac{\alpha \cdot n^s}{M_s}, \frac{n^{s-1/s}}{M_s^{1-1/s}} \right\}\right)$ as desired.	
\end{proof}

\end{document}